\documentclass[12pt]{article}

\usepackage{mathptmx}
\usepackage[english]{babel}
\usepackage[T1]{fontenc}

\usepackage{amssymb,amsmath,amsfonts,amsthm,mathrsfs}
\usepackage{stmaryrd}
\usepackage[top=2cm, bottom=2cm, left=1cm, right=1cm]{geometry}
\usepackage{authblk}
\usepackage[usenames,dvipsnames]{color}

\newtheorem{thm}{Theorem}

\newtheorem{lem}[thm]{Lemma}
\newtheorem{prop}[thm]{Proposition}

\newcommand{\R}{{\mathord{\mathbb R}}}
\newcommand{\Sp}{{\mathord{\mathbb S}}}
\renewcommand{\L}{{\mathcal{L}}}
\renewcommand{\|}{{\Vert}}
\newcommand{\V}{\ensuremath{\vec v}}
\newcommand{\W}{\ensuremath{\vec w}}
\newcommand{\vxi}{\ensuremath{\vec \xi}}
\newcommand{\veta}{\ensuremath{\vec \eta}}

\newcommand{\G}{\Gamma}

\newcommand{\Xint}[1]{\mathchoice
{\XXint\displaystyle\textstyle{#1}}%
{\XXint\textstyle\scriptstyle{#1}}%
{\XXint\scriptstyle\scriptscriptstyle{#1}}%
{\XXint\scriptscriptstyle\scriptscriptstyle{#1}}%
\!\int}
\newcommand{\XXint}[3]{\setbox0=\hbox{$#1{#2#3}{\int}$ }
\vcenter{\hbox{$#2#3$ }}\kern-.6\wd0}
\newcommand{\dashint}{\Xint-}

\begin{document}
 
\title{Uniform Approximation of a Maxwellian Thermostat by Finite 
Reservoirs}
\author[1]{F. Bonetto}
\author[1]{M. Loss}
\author[1]{H. Tossounian}
\author[1]{R. Vaidyanathan}
\affil[1]{School of Mathematics, Georgia Institute of Technology, Atlanta}

\maketitle

\abstract{We study a system of $M$ particles in contact with a large but finite 
reservoir of $N>>M$ particles within the framework of the Kac master equation 
modeling random collisions.  The reservoir is initially in equilibrium at 
temperature $T=\beta^{-1}$. We show that for large $N$, this evolution can be 
approximated by an effective equation in which the reservoir is described by a 
Maxwellian thermostat at temperature $T$. This approximation is proven for a 
suitable $L^2$ norm as well as for the Gabetta-Toscani-Wennberg (GTW) distance 
and is {\it uniform in time}.}

\section{Introduction}\label{Intro}

In \cite{Kac}, Kac studied a spatially homogeneous gas of $M$ particles moving 
in one dimension and interacting through random collisions. After certain exponentially distributed time intervals,
 a pair of particles is randomly and uniformly 
selected and they undergo a random collision, i.e., their pre-collisional 
velocities are replaced by new velocities that are randomly and uniformly 
selected in such a way that the total energy is preserved. The intensity of the 
collision process is chosen so that the average time $\lambda^{-1}$ between two 
successive collisions of a given particle, i.e., the {\it mean free time},  is 
independent of the number of particles. Thus, the $M\to\infty$ limit of the 
model can be thought of as a realization of the classical Grad-Boltzmann limit. 

To keep the presentation simple we describe the Kac model first for the system 
of $M$ particles only and deal with the full model afterwards. The sub- and 
superscript $S$ refers to this system of $M$ particles. For a spatially 
homogeneous gas the state of the system is given by a function $f(\V)$, the 
probability density of finding the particles in the system with velocities 
$\V=(v_1,\ldots,v_M)$.  The infinitesimal generator of this evolution is given 
by (see \cite{CCL,Kac})
\begin{equation}\label{Kac}
\L_S[f]=\frac{\lambda_S}{M-1}\sum_{i<j}(R_{i,j}^S-I)[f]
\end{equation}
where $I$ is the identity operator and $R_{i,j}^S$ describes the result of a 
collision between particle $i$ and particle $j$, that is
\begin{equation}
R_{i,j}^S[f](\V):=\dashint f(\V_{i,j}(\theta))d\theta 
\end{equation}
with
\begin{eqnarray}\label{rot}
&\V_{i,j}(\theta):=(v_1,\ldots,v_i^*(\theta),\ldots,v_j^*(\theta),
\ldots , v_M)\crcr
&v_i^*(\theta) := v_i \cos{\theta} + v_j 
\sin{\theta}\qquad \,\qquad v_j^*(\theta) 
:= -v_i \sin{\theta} + v_j \cos{\theta}  \ ,
\end{eqnarray} 
and
\[
 \dashint f(\theta)d\theta:=\frac1{2\pi}\int_0^{2\pi} f(\theta)d\theta.
\]

The gain term $\frac{\lambda_S}{M-1}R_{i,j}$ in \eqref{Kac} implies that, in an 
interval of length $dt$, there is a probability $\frac{\lambda_S}{M-1}dt$ that 
particles $i$ and $j$ will collide with resulting velocities $v_i$ and $v_j$. 
Because every particle label appears exactly $M-1$ times in \eqref{Kac}, 
particle $i$ has a probability $\lambda_S dt$ of being involved in a collision 
during the time interval $dt$. Thus, on average, the time between two collisions 
involving particle $i$ is $\lambda_S^{-1}$. Since the above evolution is 
completely independent of the positions of the particles, and hence of their 
density, the mean free time is the only number of physical significance.

In \cite{BLV} a Kac-type model was introduced with the additional feature that, 
besides the pair collisions, each particle  in the system can interact with a 
thermostat. The interaction of particle $j$ with the Maxwellian thermostat is 
given by

\begin{equation}\label{B}
B_j[f](\V):= \int dw \dashint d\theta \sqrt{\frac \beta {2\pi}}
e^{-\frac \beta 2 w_j^{*2}(\theta)} f(\V_j(\theta,w)) 
\end{equation}
where
\begin{equation}\label{rott}
\V_j(\theta,w)=(v_1,..., v_j \cos{(\theta)} + w 
\sin{\theta},...,v_M), \qquad\qquad w_j^*(\theta) = -v_j \sin{\theta} + w 
\cos{\theta} \ .
\end{equation}

As before, the interaction times with the thermostat are described by a Poisson 
process whose intensity $\mu$ is chosen so that the average time between two 
successive interactions of a given particle with the thermostat is independent 
of the number of particles in the system $S$. Thus, the time evolution for this 
model is given by 

\begin{equation}\label{gent}
 \dot f= \widetilde \L[f]=\L_S[f]+\widetilde\L_T[f] \ ,
\end{equation}
where
\begin{equation}\label{thermo}
\widetilde \L_T[f]=\mu 
\sum_{j=1}^{M}{(B_j-I)}[f].
\end{equation}

In order to facilitate the discussion we will call this model the {\it Thermostated System}
or T-system in short.
The unique equilibrium distribution of this thermostated system is given by a 
Gaussian with inverse temperature $\beta$. In \cite{BLV} it is shown that the 
evolution approaches this equilibrium exponentially fast in $L^2$ as well as in 
entropy {\it uniformly in} $M$. Moreover, propagation of chaos 
\cite{McK} holds for this system as well and, as $M \to \infty$, 
the evolution of the single particle marginal is given by a  Boltzmann-type 
equation. These results have been extended to a system where only a subgroup of 
the particles interact with the thermostat in \cite{TV}.

The thermostat can be thought of as an infinite reservoir of particles at a 
fixed inverse temperature $T=\beta^{-1}$ in which every particle in the 
reservoir collides at most once with a particle in the system. Thus, 
$B_j[f](\V)$ describes a collision between a system particle and a reservoir 
particle that is randomly drawn from a Maxwellian distribution with temperature 
$\beta^{-1}$. The reservoir is not affected by the collisions with the particles 
from the system $S$. If the system $S$ interacts, instead, with a large but 
finite reservoir the reservoir does not remain in equilibrium. Particles in the 
reservoir can re-collide with system particles and with other reservoir 
particles, pushing more reservoir particles out of equilibrium.

In the present paper we compare, in appropriate metrics, the evolution \eqref{gent} 
with the evolution arising from the interaction of the system $S$ with a large 
but finite reservoir $R$ containing $N>>M$ particles.  This model is explained 
in Section \ref{MR}. In Section \ref{results} we state the main results of the 
paper, namely, that for $N$ large this evolution stays close uniformly in time 
to the one with an infinite reservoir.  Section \ref{Pr} contains the proofs of 
our results. Section \ref{conc} further addresses the relevance of our results 
together with possible extensions. Finally, in the Appendices, we report some 
technical computations and discuss the optimality of our bounds.

\section{A Model for a Finite Heat Reservoir}\label{MR}

The evolution inside the reservoir $R$ is also given by a standard Kac model. As 
above, we assume that the average time between two collisions between two 
particles in the reservoir $R$ is fixed independently of $N$. We denote this 
time by $\lambda_R^{-1}$. Thus, the generator of the evolution of the reservoir 
is
\begin{equation}\label{LR}
 \L_R[f]=
 \frac{\lambda_R}{N-1}\sum_{1\leq i<j\leq N} (R^R_{i,j}-I)[f] \ .
\end{equation}
Again, the quantities that refer to the reservoir have a sub- or superscript 
$R$. The evolution of the system $S$ and the reservoir $R$ {\it without 
interaction between the two} is determined by the generator
\begin{equation}\label{LK}
 \L_K[f]=\L_S[f]+\L_R[f]
\end{equation}
where $\L_S[f]$ is given by \eqref{Kac}. The velocities of the particles in the 
system $S$ are, as before, denoted by $v_1, \dots, v_M$ and the velocities of 
the particles in the reservoir by $w_1, \dots, w_N$. Similar to what we wrote 
before, $R^S_{i,j}$ describes a collision in the system $S$ between particle 
$i$ and $j$, and is given by (see \eqref{rot})
\[ 
R^S_{i,j}[f](\V,\W) := \dashint f(\V_{i,j}(\theta),\W)d\theta 
\]
and $R^R_{i,j}$ describing a collision in the reservoir between particle $i$ and 
$j$ is  written as \[
R^R_{i,j}[f](\V,\W) := 
\dashint f(\V,\W_{i,j}(\theta))d\theta
\]
with $\V_{i,j}(\theta)$ defined in \eqref{rot} and $\W_{i,j}(\theta)$ 
analogously defined.

Some thought has to be given to the modeling of the interaction between the 
system $S$ and the reservoir $R$. Naturally, we want that the average time 
between two successive collisions of a {\it given particle in the system $S$ 
with any particle in the reservoir $R$} to be fixed independently of $N$ and 
$M$. 
This is achieved by defining the interaction generator as
\begin{equation}\label{intera}
 \L_I[f]=\frac{\mu}{N}\sum_{i=1}^M\sum_{j=1}^N (R^I_{i,j}-I)[f]
\end{equation}
where
\[ 
R^I_{i,j}[f](\V,\W) :=\dashint f(\V_{i}(\theta),\W_j(\theta))d\theta,
\]
with
\begin{align}
&\V_{i}(\theta):=(v_1,\ldots,v_i^*(\theta),\ldots, 
v_M)\qquad \,\qquad \W_{j}(\theta):=(w_1,\ldots,w_j^*(\theta),\ldots, 
w_N)\crcr
&v_i^*(\theta) := v_i \cos{\theta} + w_j 
\sin{\theta}\qquad \,\qquad w_j^*(\theta) 
:= -v_i \sin{\theta} + w_j \cos{\theta}.
\end{align}

Thus, the evolution equation for the combined system $S$ and reservoir $R$ is given by 
\begin{equation}\label{genr}
 \dot f=\L[f]=\L_K [f]+\L_I [f] \ ,
\end{equation}
where  $f$ is a probability distribution in $L^1(\R^M\times\R^N)$. It is elementary so see that
this property is preserved under the evolution \eqref{genr}. We will call this model the {\it Finite Reservoir System} or FR-system in short.

It is plain that for an arbitrary initial distribution $f_0(\V, \W)$ the 
evolutions given by \eqref{genr} and \eqref{gent} need not be similar. The 
latter tends to an equilibrium given by Gaussian at temperature $\beta^{-1}$ 
whereas the former, as can be easily seen, tends to an equilibrium which is 
given by averaging $f_0(\V,\W)$ over all rotations in $\R^{M+N}$. Clearly, there 
is no reason why these two equilibria are close in any sense.  The choice of  
initial conditions plays a key role. We shall assume that initially the 
reservoir is in the canonical equilibrium at temperature $T=\beta^{-1}$, that 
is, the state of the reservoir is given by
\[
 \G_{\beta,N}(\W)=\prod_{i=1}^{N} \Gamma_{\beta,1}(w_i)
\qquad\hbox{where}\qquad
 \G_{\beta,1}(w)=\sqrt{\frac{\beta}{2\pi}}e^{-\frac{\beta}2 
w^2} \ .
\]
We assume that the system $S$ is initially  in a generic initial state $l_0(\V)$ with
$\int l_0(\V)d\V=1$.

It is easy to see that if the total momentum is initially zero, it remains zero 
for all times. Hence, we set it equal to zero.  Moreover, we assume that the 
average kinetic energy per particle in the system is finite.
The particles are assumed to be indistinguishable so that $l_0(\V)$ is invariant 
under permutation of its variables. This implies that
\[
  \int v_i l_0(\V)d\V=0 \qquad \int |v_i|^2 l_0(\V)d\V=E_2<\infty \qquad 
\forall i.
\]
Finally, by a simple rescaling of the velocities, we can assume without loss 
of generality that $\beta=2\pi$. Thus, the initial distribution of the 
{\it system plus reservoir}, is given by
\begin{equation}\label{ini}
 f_0(\V,\W)=l_0(\V)\G_N(\W).
\end{equation}
where $\G_N(\W)=\G_{2\pi,N}(\W)$.

The evolution given by $\widetilde \L$, defined in \eqref{gent}, does not act on 
the $\W$ variables and with a slight abuse of notation we will consider 
$\widetilde \L$ as an operator acting on functions $f(\V,\W)$ of both $\V$ and 
$\W$, leaving the dependence on $\W$ unchanged. It will be sometimes convenient 
to replace the generator $\widetilde \L$ by $\widetilde \L + \L_R$. This 
substitute is legitimate, since the operator $\L_R$ leaves the reservoir at 
equilibrium.

The similarity of the two evolutions, the one given by \eqref{genr} with the one 
in \eqref{gent} acting on the same initial state \eqref{ini},  can be 
heuristically understood as follows. The form of the interaction term implies 
that, in contrast to the collisions between system particles,  the mean time 
between two successive collisions of a {\it given particle in the reservoir $R$ 
with any particle in the system $S$} is $\mu^{-1}N/M$ and thus it diverges with 
$N$. This implies that for a finite time $t$ and for $N$ very large, with 
respect to $t$, we can indeed assume that each particle in the reservoir 
collides at most once with a particle in the system. This idea is implemented 
through the choice of \eqref{B}. Thus, it is not difficult to prove a 
convergence result for any fixed time $t$, as $N\to\infty$. The interesting 
point, however, is that over longer times re-collisions will occur. Moreover the 
interaction $\L_R$, the collisions among the particles in the reservoir, spreads 
the modification of the distribution of one particle to all the reservoir 
particles. Thus, after a time approaching $N$, we can no more think that a 
randomly selected particle from the reservoir has a Maxwellian distribution. 
Thus, the real issue is to understand these competing effects in order to obtain 
a result uniformly in time. From a physical point of view such a result can be 
expected, because the thermostat is introduced to drive the system as $t \to 
\infty$ to a particular equilibrium state. 

\section{Results} \label{results}


We will always assume that the initial state $f_0$ for FR-system is of 
the form \eqref{ini}, that is, the system $S$ is in a generic initial state 
while the reservoir $R$ is in equilibrium at inverse temperature $\beta=2\pi$. The state at time $t$
of the FR-system is given by
\[
 f_t=e^{\L t}f_0 \ .
\]
As noted above, $f_t$ reaches a {\it steady state} $f_{\infty}$ when 
$t\to\infty$ and that we get:
\begin{equation}\label{ssr}
f_{\infty}(\V,\W)= \lim_{t\to 
\infty}f_t(\V,\W)=\int_{\Sp^{M+N-1}(r)} 
l_0(\V') \G_N(\W') d \sigma_r(\V',\W')
\end{equation}
where $r=\sqrt{|\V|^2+|\W|^2}$ and $\sigma_r(\V,\W)$ is the normalized uniform 
measure on the sphere of radius $r$ in $\R^{M+N}$.

We want to compare the evolution generated by $\L$ with the evolution generated by 
$\widetilde \L$, the generator for the T-system
(see \eqref{gent}). In order for them to be comparable, we think of $\widetilde \L$ as acting on 
functions of $M+N$ variables. Given an initial state $f_0$ of the form 
\eqref{ini}, let
\[
 \tilde f_t=e^{\widetilde \L t}f_0
\]
be the state of the T-system at time $t$, where clearly we have 
$\tilde f_t(\V,\W)=l_t(\V)\Gamma_N(\W)$. Any comparison between 
$f_t$ and $\tilde f_t$ will naturally yield an estimate on how much the reservoir 
deviates from its initial equilibrium state. 
Because $\L_R\Gamma_N=0$, for an initial state $f_0$ of the form \eqref{ini}, 
we can write (see \eqref{LK})
\[
 \widetilde \L=\widetilde \L_T+\L_K.
\]
This modification clearly does not change the evolution of $f_0$, but simplifies some of the 
computations below.
As $t \to \infty$,  $\tilde f_t$ approaches a steady state $\tilde f_{\infty}$  given by
\begin{equation}\label{sst}
 \tilde f_{\infty}(\V,\W)=\lim_{t\to \infty}\tilde f_t(\V,\W)=\G_{M+N}(\V,\W).
\end{equation}
It is worth observing that \eqref{ssr} and \eqref{sst} remain valid even when 
$\lambda_R=\lambda_S=0$. 

%
%

As a first attempt given in Sec. \ref{L2}, we will compare the above evolutions in the space 
$L^2(\mathbb{R}^M\times \mathbb{R}^N,\Gamma_{M+N})$. 
Since $f_0$ is a probability distribution, such an $L^2$ norm is not very 
natural, however, the computations are relatively simple. After discussing the limitations of the results in $L^2$ , we will, in Sec.\ref{GTW}, compare the evolutions in the 
Gabetta-Toscani-Wennberg (GTW) metric (see \cite{GTW}). This metric is more
natural but the computations are quite difficult.

\subsection{Evolution in $L^2(\mathbb{R}^{M+N},\Gamma_{M+N})$}\label{L2}

As discussed in \cite{BLV}, it is natural to look at the evolution in the 
ground state representation by defining
\[
f_t(\V,\W) = h_t(\V,\W)\G_{M+N}(\V,\W)
\]
where
\[
f_0(\V,\W) = h_0(\V)\G_{M+N}(\V,\W)
\]
with $\int h_0(\V)\G_N(\V) d\V=1$ while $\int v_i h_0(\V)\G_N(\V) d\V=0$ 
and $\int |v_i|^2 h_0(\V) \G_N(\V)d\V=E_2$, for every $i$.

Observe that $\L_K$ (see \eqref{LK}) has the same form when 
acting on $f$ or on $h$. More precisely we have that
\[
 \L_K[\Gamma_{M+N}h]=\Gamma_{M+N}\L_K[h].
\]
This easily follows from the fact that $\Gamma_{M+N}$ is a rotationally invariant 
function. On the other hand, in the case of the thermostat we 
have to note that
\[
 B_i[\Gamma_{M+N}h]=\Gamma_{M+N}T_i[h]
\]
where $B_i$ is given by \eqref{B} while
\begin{equation}\label{Ti}
 T_i[f]=\int dw  
e^{-\pi w^{2}} \dashint f(\V_i(\theta,w))d\theta \ .
\end{equation}
This means that the evolution of the initial state $h_0$ under the 
thermostated evolution can be written has
\[
 \tilde h_t=e^{\overline \L t}h_0
\]
where
\[
\overline \L[h]=\L_K[h]+\L_T[h]
\]
with
\[
\L_T[h]=\mu\sum_{i=1}^M(T_i-I)[h]\ .
\]
Recall that $\L_S+\L_T$ acts only on the $\V$ variables while $\L_R$ acts only 
on the $\W$ variables. Thus, if $h_0$ depends only on $\V$ then $e^{\overline \L 
t}h_0$ will depend only on $\V$ too. It follows that the term $\L_R$ is 
identically zero along the evolution of the chosen initial state. We keep it for
future comparison with $\L$. Note that 
$\L[h\G_{M+N}]=\L[h]\G_{M+N}$ and hence the generator of the evolution for the 
FR-system requires no modifications.

It is easy to see that $\L$ and $\overline \L$ are bounded self-adjoint operators 
on $L^2(\mathbb{R}^{M+N},\Gamma_{M+N})$ with the scalar product
\begin{equation}\label{sp}
 \langle f,g\rangle=\int f(\V,\W)g(\V,\W)\Gamma_{M+N}(\V,\W)d\V d\W.
\end{equation}
Thus, it is natural to assume that $h_0\in L^2(\R^{M+N},\G_{M+N}(\V,\W))$ and  
to study the evolution of $\|e^{\overline\L t}h_0-e^{\L t}h_0\|_2$.

As a first step we estimate the behavior of the difference of the steady states.
We clearly have
\[
f_{\infty}(\V,\W)=\G_{M+N}(\V,\W) h_{\infty}(\V)
\]
with
\[
 h_{\infty}(\V,\W)=\int_{\mathbb{S}^{M+N-1}(r)}h(\V)d\sigma_r(\V,\W)
\]
whereas $\tilde h_{\infty}\equiv 1$. In Appendix \ref{sL2eq}, we show that

\begin{equation}
\label{L2eq}
\|h_{\infty}-\tilde h_{\infty}\|_2^2=\int_{\R^{M+N}} [h_{\infty}(\V,\W)-1]^2 
\Gamma_{M+N}(\V,\W) d\V d\W 
\leq \frac{M}{N-2} \vert\vert h_0 - 1 \vert\vert_2^2
\end{equation}
Thus, the distance between the steady states is controlled by the distance 
between the initial state and the canonical equilibrium state and it vanishes 
as $1/ \sqrt N$ as $N\to \infty$. This estimate, in a slightly weaker form, 
remains true for all $t$. 

\begin{thm}\label{thL2}
 Let $f_0$ be the initial distribution for the system with reservoir and assume 
that it has the form
\begin{equation}\label{ini2}
 f_0(\V,\W)=h_0(\V)\Gamma_{M+N}(\V,\W)
\end{equation}
with $h_0\in L^2(\R^{M+N},\G(\V,\W))$. Then for every $t>0$ we have
\begin{equation}\label{estL2}
 \|e^{\overline\L t}h_0-e^{\L t}h_0\|_2\leq \frac{M}{\sqrt 
N}(1-e^{-\frac\mu 2 t})\Vert 
h_0-1\Vert_2\ .
\end{equation}
\end{thm}
This statement is proved in Section \ref{pL2ev}. 

\if
The main 
observation in the proof is that, on a distribution of the form \eqref{ini2}, 
the actions of $\widetilde\L_T$ (see \eqref{thermo}) and $\L_I$ (see 
\eqref{intera}) are very similar, since the particles in the reservoir have a 
Maxwellian distribution. Clearly, as already observed, the form \eqref{ini2} is 
not preserved by the evolution. We will thus introduce a suitable expansion of 
the difference of the evolutions that will allow us to exploit the above 
observation at all times.
\fi
We close this section with some remarks about the meaning of Theorem \ref{thL2}.
In view of the estimate on the steady states, we see that the dependence on $N$ 
in \eqref{estL2} is optimal. Observe that the particles in the  
reservoir of the FR-model are at thermal equilibrium at time 0 and then evolve to a radially 
symmetric state for large time. Hence it is not surprising that the final 
state is close to a canonical distribution. Thus, the fact the 
their state remains close to a canonical distribution {\it uniformly in time} is the 
main point of the above theorem.

Observe that the dependence of the estimate on $M$ during the evolution is 
not the same as in the steady state. It is not clear to us whether this is an 
artifact of our proof. The main ingredient in the proof is the estimate 
\eqref{middleL2}. In Appendix \ref{M}, we show that this estimate is optimal in 
its $M$ behavior. This implies that the time derivative at $t=0$ of 
$\|e^{\overline\L t}h_0-e^{\L t}h_0\|$ can actually be $M/\sqrt N$. But this may only 
be true for a very small time.

A disturbing aspect of the theorem is that it behaves very poorly when applied 
to some very reasonable initial distributions. Assume that the system is 
initially in equilibrium at a temperature $T_S=\beta_S^{-1}\not=\beta^{-1}$, 
that is $f_0(\V)=\Gamma_{\beta_S,M}(\V)\Gamma_{\beta,M}(\W)$. It follows that 
$h_0(\V)=\Gamma_{\beta_S,M}(\V)/\Gamma_{\beta,M}(\V)$. If $2\beta_S\geq\beta$ then $\Vert 
h_0\Vert_2=C(\beta_S)^M$ where 
$C(\beta_S)^2=\beta_S/\sqrt{\beta(2\beta_s-\beta)}>1$. Thus, if the right hand side of \eqref{estL2} is to be small for such an initial 
state, we need a 
reservoir with a number of particles $N$ exponentially large in $M$. In a sense, 
this makes the behavior in $M$ discussed above rather unimportant. Also, if the 
initial temperature is sufficiently large, that is if $2\beta_S\leq\beta$, then 
$C(\beta_S)=\infty$, $h_0\not\in L^2(\mathbb{R}^M,\Gamma_M(\V))$ and our theorem
does not apply in this situation.
These are, perhaps, the main reasons why the Gabetta-Toscani-Wennberg metric is 
better suited for our purposes, although it is quite a bit more difficult to 
handle. 

\subsection{The Gabetta-Toscani-Wennberg metric}\label{GTW}

The Gabetta-Toscani-Wennberg (GTW) 
metric is a distance between probability densities.  Let $f,g\in L^1(\mathbb{R}^{M+N})$ be two possible distributions for the FR-system where
\begin{equation}\label{condd2}
\int v_i 
f(\V,\W)d\V d\W=\int w_j f(\V,\W)d\V d\W=0\qquad
\int v_i^2 f(\V,\W)d\V d\W,\int w_j^2 
f(\V,\W)d\V d\W<\infty
\end{equation}
and analogously for $g$. We can define then
\begin{equation}\label{dd2}
d_2(f,g) := \sup_{\vxi \not= 0, \veta \not= 0} \frac{|\widehat f 
(\vxi,\veta) - 
\widehat g(\vxi,\veta)|}{|\vxi|^2+|\veta|^2}.
\end{equation}
Here, and in the following, we use the convention that $\widehat f$, 
the Fourier transform of $f$, is given by
\[
 \widehat 
f(\vxi,\veta)=\int_{\R^{M+N}}e^{-2\pi i(\vxi,\V)}e^{-2\pi i(\veta,\W)}f(\V,\W)d\V d\W,
\]
where $\vxi=(\xi_1,\ldots,\xi_M)$ are the Fourier variables associated with the 
particles in the system $S$, while $\veta=(\eta_1,\ldots,\eta_N)$ are the 
Fourier variables associated with the particles in the reservoir $R$. It is easily seen that
under the stated conditions, $d_2(f,g)$ is defined. The 
metric $d_2$ in \eqref{dd2} is the more interesting member of a family of 
metrics 
$\{d_\alpha\}$ introduced in \cite{GTW}.

Again we imagine that our system starts at time 0 in a state of the form 
\[
f_0(\V,\W)=l_0(\V)\G_N(\W)
\]
and we want to estimate the $d_2$ distance between $f_t=e^{\L t}f_0$ 
and $\tilde f_t=e^{\widetilde \L t}f_0$. To see what kind of behavior to expect, 
we start from the distance between the steady states.
Because the Fourier transform commutes with rotations we find
\[
\widehat f_{\infty}(\vxi,\veta) = \int_{\Sp^{M+N-1}(r)} \widehat l_0(\vxi) 
\G_N(\veta)  d\sigma_r(\vxi,\veta) 
\]
and
\[
\widehat{\tilde f}_{\infty}(\vxi,\veta)=\G_{M+N}(\vxi,\veta)
\]
where we have used that $\Gamma_1$ is invariant under the Fourier transform.
In Appendix \ref{sd2eq}, we show that

\begin{equation}\label{d2eq}
d_2(f_{\infty},\tilde f_{\infty}) \le  \frac{M}{ M+N } 
d_2( l_0, \Gamma_M).
\end{equation}

Again we want to obtain an estimate that remains true uniformly in time. In 
Section \ref{pd2ev}, we prove the following. 

\begin{thm}\label{thd2}
 Let $f_0(\V,\W)$ be the initial distribution for the system plus reservoir of 
the form
\[
 f_0(\V,\W)=l_0(\V)\G_N(\W).
\]
with $l_0$ symmetric and satisfying \eqref{condd2}.
Assume moreover that the fourth moment
\begin{equation}\label{fourth}
\int v_i^4 l_0(\V)d\V=E_4<\infty \ .
\end{equation}
Then for every $t>0$ we have
\begin{equation}\label{estd2}
 d_2\left(e^{\widetilde\L t} f_0, e^{\L t} f_0\right)\leq 
\frac{KM}{N}\left(1-e^{-\frac\mu 4 
t}\right)\sqrt{d_2(l_0,\Gamma_M)(F_4+d_2(l_0, \Gamma_M)) }\, .
\end{equation}
with $F_4=48\pi^4 (E_4+1)$ and $K=16\sqrt{2}$.
\end{thm}

The basic strategy of the proof of this theorem is similar to the one used for the proof of Theorem 
\ref{thL2}. Having said this, estimating the difference between 
$\widetilde \L_T$ and $\L_I$ in the $d_2$ metric turns out to be considerably more 
difficult than the one in the $L^2$ norm. Most of the work in the proof of 
Theorem \ref{thd2} in Section \ref{pd2ev} is devoted to carrying out these estimates which are summarized in Proposition 
\ref{estim}. It is really in the proof of Proposition \ref{estim} that the extra 
condition \eqref{fourth} on the fourth order moment of the initial 
distribution is needed. In Appendix \ref{M} we show that such a condition is 
indeed necessary for our proof.

We observe that $d_2(l_0,\Gamma_M)$ is well defined for any $l_0$ satisfying 
\eqref{condd2}. Moreover, if $l_0$ is a product state, that is if 
\[
l_0(\V)=\prod_{i=1}^M \ell(v_i)           
\]
then, calling 
$\vxi^{<i}=(\xi_1,\ldots,\xi_{i-1})$, $\vxi^{>i}=(\xi_{i+1},\ldots,\xi_M)$ and 
$\widehat l^{>i}_0(\vxi^{>i})=\prod_{j>i} \widehat \ell(v_j)$, we get   
\[
 \frac{|\Gamma_M(\vxi)-\widehat l_0(\vxi)|}{|\vxi|^2}\leq
\frac{\sum_i 
\Gamma_{i-1}(\vxi^{<i})\,\left|\Gamma_1(\xi_i)-\widehat\ell(\xi_i)\right|\,
\widehat l^{>i}_0(\vxi^{>i})}{\sum_i \xi_i^2}\leq
\sup_i \frac{|\Gamma_1(\xi_i)-\widehat\ell(\xi_i)|}{\xi_i^2}
\]
so that
\[
 d_2(l_0,\Gamma_M)=d_2(\ell,\Gamma_1).
\]
These observations address both problems found in the $L^2$ estimate.

\section{Proof of Theorem \ref{thL2} and Theorem \ref{thd2}}\label{Pr}

Both proofs are based on an expansion of the difference between two 
exponentials that we discuss here in the form needed for the $L_2$ estimates. A 
very similar expansion can be obtained for the $d_2$ case.

Observe that we can write
\begin{align}
 \L=&Q_S+Q_R+Q_I-\Lambda I\crcr
 \overline \L=&Q_S+Q_R+Q_T-\Lambda I
\end{align}
where 
\[
\Lambda=\frac{\lambda_S}{2}M+\frac{\lambda_R}{2}N+\mu M
\]
while
\[
 Q_S=\frac{\lambda_S}{M-1}\sum_{1\leq i<j\leq M} R^S_{i,j}\qquad\qquad 
Q_R=\frac{\lambda_R}{N-1}\sum_{1\leq i<j\leq M} R^R_{i,j}\ .
\]
Finally,
\[
 Q_I=\frac{\mu}{N}\sum_{i=1}^M\sum_{j=1}^N R^I_{i,j}\qquad\qquad 
Q_T=\mu\sum_{i=1}^M T_i \ .
\]
We can thus write
\[
 e^{\L t}-e^{\overline \L t}=e^{-\Lambda t}\sum_{n=1}^\infty 
\frac{t^n}{n!}\left[ (Q_S+Q_R+Q_I)^n-(Q_S+Q_R+Q_T)^n\right].
\]
We further expand each term in the above sum as
\[ 
(Q_S+Q_R+Q_I)^n-(Q_S+Q_R+Q_T)^n=\sum_{k=0}^{n-1}
(Q_S+Q_R+Q_I)^{n-1-k}(Q_I-Q_T)(Q_S+Q_R+Q_T)^k 
\]
so that we get
\begin{equation}\label{expaL2}
 e^{\L t}-e^{\overline \L t}=e^{-\Lambda t}\sum_{n=1}^\infty 
\frac{t^n}{n!}\sum_{k=0}^{n-1}
(Q_S+Q_R+Q_I)^{n-1-k}(Q_I-Q_T)(Q_S+Q_R+Q_T)^k\ .
\end{equation}
The above expansion has three major advantages:

\begin{enumerate}
 \item Isolating the factor $e^{-\Lambda t}$ avoids expanding a negative 
exponential as a power series.

 \item As discussed in the previous section, we expect the difference between 
$Q_I$ and $Q_T$ to be small when they act on a function that depends only on 
$\V$. It is easy to see that $h_k(\V):=(Q_S+Q_R+Q_T)^k h_0(\V)$ still depends 
only on $\V$ so that we expect to gain from the term $(Q_I-Q_T)h_k$.

\item Finally $\Lambda$ is the largest eigenvalue of $Q_S+Q_R+Q_T$ 
corresponding to the eigenvector 1. But $(Q_I-Q_T)1=0$ so that, writing 
$h_k=1+u_k$, we expect that $\Vert u_k\Vert_2<\Lambda^k$. A uniform version of 
this estimate, see \eqref{beforeL2} below, allows us perform the sum over 
$k$ in \eqref{expaL2} without paying a factor of $n$. This is crucial in obtaining a 
bound uniform in $t$. 
\end{enumerate}
The following proofs consist, to a large extent, in a quantitative 
implementation of the above three observations.

\subsection{Proof of Theorem \ref{thL2}}\label{pL2ev}

Observe that  $(e^{\L t}-e^{\overline \L t}) 1\equiv 0$ because the 
constant function $1$ is a steady state for both evolutions. For this reason, 
we will write
\[
 h_0(\V)=1+u_0(\V)\qquad\hbox{with}\qquad \langle u_0,1\rangle_M=0
\]
where $\langle\cdot,\cdot\rangle_M$ is the scalar product in 
$L^2(\mathbb{R}^M,\G_M(\V))$, that is
\[
 \langle u,h \rangle_M=\int u(\V)h(\V)\G_M(\V)d\V\ .
\]
From now on we will identify $L^2(\mathbb{R}^M,\G_M(\V))$ with a subspace of 
$L^2(\mathbb{R}^{M+N},\G_{M+N}(\V,\W))$. We thus need to estimate the norm of
\[
 (Q_S+Q_R+Q_I)^{n-k-1}(Q_I-Q_T)(Q_S+Q_R+Q_T)^k u_0(\V) \ .
\]
To this end, observe that $R^S_{i,j}$ is the orthogonal projector on the 
subspace of functions that are invariant under rotations of $v_i$ and $v_j$ so 
that
\[
\|R^\alpha_{i,j}\|_2 = 1 \qquad\hbox{ for }\qquad \alpha=S,\ R\ \hbox{or}\ I, 
\]
while 
\[
\|Q_T u\|_2\leq \mu\left( M-\frac12\right)\|u\|_2\qquad\hbox{ if }\qquad 
\langle 
u,1\rangle=0.
\]
Observe indeed that $Q_T$ is a sum of operators 
acting independently on each variable $v_i$. Thus, its eigenvectors are tensor 
products 
of the eigenvectors of each of the $T_i$, while its eigenvalues are sums of 
their eigenvalues. It is possible to see that the Hermite polynomial 
$H_{2n}(v_i)$ of degree $2n$ and weight $e^{-\pi v_i^2}$ is an eigenvector of 
$T_i$ with eigenvalue $a(n)$. The last inequality then follows from the fact 
that $a(0)=1$ is the 
largest eigenvalue of $T_i$ with eigenvector $H_0(v_i)=1(v_i)$, while 
$a(n)\leq 1/2$ for $n>0$. It follows that $\|T_i l\|_2\leq 
(1/2)\|l\|$ when $\langle l,1\rangle=0$. With this, we get that
\[
 \langle(Q_S+Q_R+Q_T)u,1\rangle=0 \qquad\hbox{if}\qquad \langle u,1\rangle=0
\]
and
\begin{equation}\label{beforeL2}
\|u_k\|_2\leq \left(\Lambda-\frac{\mu}{2}\right)^k\|u_0\|_2 \ ,
\end{equation}
where
\[
u_k:= (Q_S+Q_R+Q_T)^k u_0 \ ,
\]
while
\begin{equation}\label{afterL2}
 \|Q_S+Q_R+Q_I\|_2\leq \Lambda\ .
\end{equation}
We thus have to estimate $\|(Q_I-Q_T)u\|_2$ where $u$ depends only on $\V$.

\begin{lem}\label{fl} Let $u(\V) $ be any function in $L^2(\R^M, 
\Gamma_M(\V))$. 
Then
\[
\left\Vert \frac{1}{N} \sum_{j=1}^N R^I_{i,j} u - T_i u \right\Vert^2_2  = 
\frac{1}{N}\left(\langle T_i u, u\rangle - \langle T_i u,T_i u\rangle\right)
\]
\end{lem}

\begin{proof} 
Consider for simplicity $i=1$. We get
\begin{align*}
\left\Vert\frac{1}{N} \sum_{j=1}^N R^I_{1,j} u - T_1 u \right\Vert_2^2
=&\frac{1}{N^2}  \sum_{j,k = 1}^N \int_{\R^{M+N}} R^I_{1,j} u R^I_{1,k} u 
d\mu(\V,\W) - 
\frac{2}{N} \sum_{j=1}^N \int_{\R^{M+N}} R^I_{1,j} u T_1u  d\mu(\V,\W)+\crcr
+&\int_{\R^{M+N}} |T_1u(v)|^2 d \mu(\V,\W)  \ ,
\end{align*}
where $d \mu(\V,\W)=\G_{M+N}(\V,\W)d\V d\W$. Calling 
$\V^1=(v_2,\ldots,v_M)$, we note that
\begin{align}\label{V2}
\int_{\R^{M+N}} R^I_{1,1} u T_1u d \mu(\V,\W) =&\int_{\R^{M-1}}\int_{\R^2} 
\dashint u(\sin \theta v_1 + \cos \theta w_1,\V^1) d \theta 
T_1u(\V) \G_1(v_1) \G_1(w_1) d v_1 d w_1 \G_{M-1}(\V^1)d\V^1=\crcr
=&\int_{\R^M} |T_1u(\V)|^2 \G_M(\V) d\V \ .
\end{align}
Moreover, 
\begin{align*}
\int_{\R^{M+N}} R^I_{1,1} u R^I_{1,2} u d \mu(\V,\W) 
=&\int_{\R^{M-1}}\int_{\R^3} \dashint u (\sin \theta v_1 + \cos \theta 
w_1,\V^1)d \theta\dashint u (\sin \theta v_1 + \cos \theta w_2,\V^1)d 
\theta \cdot\crcr
&\qquad\qquad\cdot\G_1(v_1) \G_1(w_1) \G_1(w_2) dv_1 dw_1 
dw_2\G_{M-1}(\V^1)d\V^1=\crcr
=&\int_\R |T_1(u)(\V)|^2 d \mu(\V,\W)\ .
\end{align*}
Finally, we observe that $R^I_{i,j}$ is a projector, so that
\[
 \int_{\R^{M+N}} R^I_{1,1} u R^I_{1,1} u d\mu(\V,\W)=\int_{\R^{M+N}}  u 
R^I_{1,1} u d\mu(\V,\W)=\int_{\R^{M+N}}  u 
T_{1} u d\mu(\V,\W)
\]
where the last equality follows as in \eqref{V2}.
Collecting all terms proves the lemma.
\end{proof}

It thus follows that
\begin{align}
 \|(Q_I-Q_T)u_k\|_2^2= \mu\left\Vert\sum_{i=1}^M\left(\frac{1}{N}\sum_{j=1}^N 
R^I_{i,j}-T_i\right)u_k\right\Vert_2^2\leq &\mu 
M\sum_{i=1}^M\left\Vert\left(\frac{1}{N}\sum_{j=1}^N  R^I_{i,j}-T_i\right) u_k 
 \right\Vert_2^2\leq\crcr
 \leq &\frac{\mu M}{N} \sum_{i=1}^M (u_k, T_i 
u_k)-(T_iu_k,T_iu_k) \ .
\end{align}
Observe that if $\langle u, 1\rangle=0$, we can write $u=\bar u+\tilde u$ where 
$\bar u$ does not depend on $v_1$ while
\[
 \int \tilde u(\V)\Gamma_1(v_1)dv_1=0\qquad\qquad \forall \V^1 \ .
\]
It follows that
\[
 \langle T_1 u, u\rangle - \langle T_1 u,T_1 u\rangle=\langle T_1 \tilde u, 
\tilde u\rangle - \langle T_1 \tilde u,T_1 \tilde u\rangle\leq
\sup_k(\rho_k-\rho_k^2)\|\tilde u\|_2
\]
where $\rho_k$ are the eigenvalues of $T_i$ different from 1. Since 
$\rho_k\leq 1/2$ (see \cite{BLV}) and $x^2-x$ is 
increasing on $[0,1/2]$, we get
\begin{equation}\label{middleL2}
 \|(Q_I-Q_T)u_k\|_2\leq\frac{\mu}{2}\frac{M}{\sqrt{N}}\|u_k\|_2.
\end{equation}
Combining \eqref{middleL2},\eqref{beforeL2} and \eqref{afterL2}, we get
\[
 \|(Q_S+Q_R+Q_I)^{n-k-1}(Q_I-Q_T)(Q_S+Q_R+Q_T)^k h_0(\V)\|_2\leq 
\frac{\mu}{2}\frac{M}{\sqrt{N}}\Lambda^{n-k-1}\left(\Lambda-\frac{\mu}{2
} \right)^k\|h_0-1\|_2 \ .
\]
Adding up, we obtain
\begin{align*}
 \left\Vert (Q_S+Q_R+Q_I)^n h_0-(Q_S+Q_R+Q_T)^n h_0\right\Vert_2\leq& 
 \frac{\mu}{2}\frac{M}{\sqrt{N}}
\Lambda^{n-1}\|h_0-1\|_2\sum_{k=0}^{n-1}\left(1-\frac{\mu 
}{2\Lambda}\right)^k=\crcr
=&\frac{M}{\sqrt{N}}
\Lambda^n\left[1-\left(1-\frac{\mu 
}{2\Lambda}\right)^n\right]\|h_0-1\|_2
\end{align*}
Thus, finally,
\begin{equation}\label{fin}
\|(e^{\L t}-e^{\overline \L t}) h_0\|_2\leq \|h_0-1\|_2
\frac{M}{\sqrt{N}}e^{-\Lambda t}\sum_{n=0}^\infty 
\frac{t^n}{n!}\Lambda^n\left[1-\left(1-\frac{\mu 
}{2\Lambda}\right)^n\right]=\|h_0-1\|_2
\frac{M}{\sqrt{N}}\left(1-e^{-\frac{\mu}{2}t}\right) \ .
\end{equation}
This concludes the proof of Theorem \ref{thL2}.

\subsection{Proof of Theorem \ref{thd2}}\label{pd2ev}

We can proceed as in eq.\eqref{expaL2} to obtain

\begin{equation}\label{expad2}
 e^{\L t}-e^{\widetilde\L t}=e^{-\Lambda t}\sum_{n=1}^\infty 
\frac{t^n}{n!}\sum_{k=0}^{n-1}
(Q_S+Q_R+Q_I)^{n-1-k}(Q_I-Q_B)(Q_S+Q_R+Q_B)^k\ .
\end{equation}
where we set as before
\[
  \widetilde \L=Q_S+Q_R+Q_B-\Lambda I
\]
with
\[
 Q_B=\mu\sum_{i=1}^M B_i \ .
\]
Using this expansion in the definition \eqref{dd2} we get
\begin{equation}\label{expa}
 d_2\left(e^{\L t}f_0,e^{\widetilde \L t}f_0\right)\leq e^{-\Lambda 
t}\sum_{n=1}^\infty 
\frac{t^n}{n!}\sum_{k=0}^{n-1}\Lambda^{k}d_2\left(
(Q_S+Q_R+Q_I)^{n-1-k}Q_I[l_k\Gamma_N],(Q_S+Q_R+Q_I)^{n-1-k}Q_B[l_k\Gamma_N]
\right) 
\end{equation}
where
\begin{equation}\label{lk}
 l_k\Gamma_N=\Lambda^{-k}(Q_S+Q_R+Q_B)^k[l_0\Gamma_N]\qquad\hbox{that is}\qquad 
l_k=\Lambda^{-k}\left(Q_S+Q_B+\frac{\lambda_R N}{2}I\right)^k[l_0]
\end{equation}
because $Q_R$ acts as a multiple of the identity on $\Gamma_N$ and $Q_B$ as well as $Q_S$ act only on $l_0$. We have 
introduced the factor $\Lambda^{-k}$ to maintain the normalization of $l_k$, 
that is $\int l_k(\V)d\V=1$.

We thus need estimates for $d_2$ that can play an analogous role as eq. 
\eqref{beforeL2}, \eqref{afterL2} and \eqref{middleL2} played in the proof of 
Theorem \ref{thL2} in section \ref{pL2ev}. 

As a first thing, we need representations of the Fourier transform of the 
collision and thermostat operators. Let $f(\V,\W)$ be a function of $(\V,\W)$. 
Since the Fourier transform commutes with rotations, we get
\[
 \widehat{R^S_{i,j}[f]}(\vxi,\veta)=\dashint d\theta \hat 
f(\xi_{i,j}(\theta),\veta):=\widehat{R^S_{i,j}}[
\hat f] (\vxi,\veta)
\]
where $\xi_{i,j}(\theta)$ is defined as in \eqref{rot}. An analogous formula holds 
for $R^I_{i,j}$ and $R^R_{i,j}$.
Moreover, we get
\[
\widehat{B_i[f]}(\vxi, \veta)=\dashint d\theta \hat 
f(\xi_i(\theta,0),\veta):=\widehat{B_i}[
\hat f ] (\vxi,\veta).
\]
The behavior of these two operators under the $d_2$ metric is contained in the 
following Lemma.

\begin{lem}\label{RB}
Let $f(\V,\W)$ and $g(\V,\W)$ be two distributions, with $0$ first moment 
and finite second moment. We have 
\begin{equation}\label{afterd2}
 d_2\left(\Lambda^{-1}(Q_S+Q_R+Q_I)f,\Lambda^{-1}(Q_S+Q_R+Q_I)g\right)\leq 
d_2\left(f,g\right)\ .
\end{equation}
Assume moreover that $f(\V,\W)=l(\V)\G_{N}(\W)$ then
\begin{equation}\label{befored2}
 d_2\left(\Lambda^{-1}(Q_S+Q_R+Q_B)f,\Gamma_{M+N}\right)\leq 
\left(1-\frac{\mu}{2\Lambda}\right)d_2\left(f,\Gamma_{M+N}\right)=\left(1-\frac{
\mu}{2\Lambda}\right)d_2\left(l,\Gamma_{M}\right)
\end{equation}
\end{lem}
\begin{proof}
It is easy to see that $d_2(f,g)$ is jointly convex in $f$ and $g$, that is for 
every $\alpha,\beta>0$ with $\alpha+\beta=1$, we have
\begin{equation}\label{conv}
 d_2(\alpha f_1+\beta f_2,\alpha g_1+\beta g_2)\leq \alpha d_2(f_1,g_1)+\beta 
d_2(f_2,g_2).
\end{equation}
We have
\[  
\widehat{R^S_{i,j}}[f](\vxi,\veta)-\widehat{R^S_{i,j}}[g]
(\vxi,\veta)=\dashint d\theta 
\left(\hat 
f(\vxi_{i,j}(\theta),\veta)-\hat g(\vxi_{i,j}(\theta),\veta)\right)
\]
and, because $|\vxi_{i,j}(\theta)|=|\vxi|$, we get
\begin{align}\label{RR} 
d_2\left(R^S_{i,j}f,R^S_{i,j}g\right)\leq&\sup_{\vxi,\veta\not=0}\frac{\dashint 
d\theta 
\left|\hat 
f(\vxi_{i,j}(\theta),\veta)-\hat 
g(\vxi_{i,j}(\theta),\veta)\right|}{|\vxi_{i,j}(\theta)|^2+|\veta|^2}
\leq\crcr
\leq&\sup_{\vxi,\veta\not=0,\theta}\frac{\left| 
\hat f(\vxi_{i,j}(\theta),\veta)-\hat 
g(\vxi_{i,j}(\theta),\veta)\right|}{|\vxi_{i,j}
(\theta)|^2+|\veta|^2 }=d_2\left(f,g\right)
\end{align}
Clearly, an identical argument holds for $R^{I}_{i,j}$ and $R^{R}_{i,j}$. 
Equation \eqref{afterd2} follows from the convexity property \eqref{conv}. 

Because 
$B_i\G_{M}=\G_{M}$ we get
\begin{align}\label{BB}
d_2\left(\frac{1}{M}\sum_{i=1}^M B_il_0, 
\G_M\right)\leq&\frac{1}{M}\sup_{\vxi\not=0}\sum_{i=1}^M\dashint \frac{ 
\left|\hat 
l(\vxi_{i}(\theta,0))-\G_M(\vxi_{i}(\theta,
0))\right|\Gamma_1(\zeta_i\sin\theta)} { 
|\vxi_i(\theta,0)|^2}\frac{\left|\vxi_i(\theta,0)\right|^2}{|\vxi|^2}
d\theta\leq \crcr
\leq &d_2\left(l,\G_M\right)\frac{1}{M}\dashint d\theta 
\sum_{i=1}^M\frac{|\vxi|^2-\xi_i^2\sin^2\theta}{|\vxi^2|}=\left(1-\frac{
\dashint d\theta \sin^2\theta}{M}\right)d_2\left(l,\G_M\right).
\end{align}
Again \eqref{befored2} follows from \eqref{conv}.

\end{proof}

Combining \eqref{expa} and \eqref{afterd2} we get
\begin{equation}\label{expad22}
 d_2\left(e^{\L t}f_0,e^{\widetilde \L t}f_0\right)\leq
 e^{-\Lambda 
t}\sum_{n=1}^\infty 
\frac{t^n\Lambda^{n-1}}{n!}\sum_{k=0}^{n-1}d_2\left(Q_I[l_k\Gamma_N], 
Q_B[l_k\Gamma_N]\right)
\end{equation}
Thus we want to estimate
\begin{equation}\label{D2D}
\frac{1}{ M}d_2(Q_I[l_k\Gamma_N],Q_B[l_k\Gamma_N]) =\frac{\mu}{MN}\sup_{\vxi,
\veta\not=0 } \frac { 1 } {
|\vxi|^2+|\veta|^2}\left|\sum_{i=1}^M\sum_{j=1}^N 
\left(\widehat{R}^I_{i,j}[\widehat 
l_k\Gamma_N](\vxi,\veta)-\widehat B_i[\widehat 
l_k\Gamma_N](\vxi,\veta)\right)\right| \ ,
\end{equation}
where $l_k$ is defined in \eqref{lk}. 
Setting
\[
\widehat F_{k,i}(\vec \xi, \eta_j)  = \dashint d\theta 
\hat l_k(\xi_1,\ldots,\xi_i\cos\theta+\eta_j\sin\theta,\ldots,
\xi^M)\Gamma_1(-\xi_i\sin\theta+\eta_j\cos\theta)
\]
we can write
\[
 \widehat{R}^I_{i,j}[\widehat 
l_k\Gamma_N] =  \Gamma_ { N-1}(\vec \eta^j) \widehat F_{k,i}(\vxi,\eta_j)
\]
where $\vec \eta^j=(\eta_1,\ldots,\eta_{j-1},\eta_{j+1},\ldots,\eta_N)$.
Likewise,
\[
 \widehat B_{i}[\widehat 
l_k\Gamma_N]=\Gamma_{N}(\eta)\dashint d\theta 
\hat l_k(\xi_1,\ldots,\xi_i\cos\theta,\ldots,
\xi^M)\Gamma_1(-\xi_i\sin\theta)=\widehat 
F_{k,i}(\vxi,0)\Gamma_{1}(\eta_j)\Gamma_{N-1}(\eta^j).
\]
Thus calling
\begin{equation}\label{bGG}
\widehat G_k(\vxi,\eta)= \frac{1}{M}\sum_{i=1}^M \left(\widehat 
F_{k,i}(\vxi,\eta)-\widehat F_{k,i}(\vxi,0)\Gamma_1(\eta)\right)
\end{equation}
we can rewrite \eqref{D2D} in a more compact form 
\begin{equation}\label{GG}
\frac1M d_2(Q_I[l_k\Gamma_N],Q_B[l_k\Gamma_N])
=\frac\mu N\sup_{\vxi,
\veta\not=0 } \frac{1}{
|\vxi|^2+|\veta|^2}\sum_{j=1}^N \widehat G_k(\vxi,\eta_j)\Gamma_{N-1}(\veta^j).
\end{equation}
Moreover, we have that
\begin{align*}
 F_{k,i}(\V,w)=&\dashint d\theta 
\hat l_k(v_1,\ldots,v_i\cos\theta+w\sin\theta,\ldots,
v^M)\Gamma_1(-v_i\sin\theta+w\cos\theta)=\crcr
=&\dashint d\theta 
\hat l_k(v_1,\ldots,v_i\cos(-\theta)-w\sin(-\theta),\ldots,
v^M)\Gamma_1(v_i\sin(-\theta)-w\cos(-\theta))=F_{k,i}(\V,-w)
\end{align*}
where we have used that $\Gamma_1$ is an even function. Thus $\widehat 
F_{k,i}(\vxi,\eta)$ is even in $\eta$ which makes $\widehat G_k(\vxi,\eta)$ 
even in $\eta$. We also have $\widehat G_k(\vxi,0)=0$.

Our goal is to bound $d_2(Q_I[l_k\Gamma_N],Q_B[l_k\Gamma_N])$ in 
terms of $d_2(l_k,\Gamma_M)$. Thus, we focus on the supremum over 
the $\veta$ variables of the reservoirs $R$, that is we look at
\begin{equation}\label{supeta}
 {\cal D}_N\left(\widehat G_k(\vxi,\cdot),|\vxi|\right)=\sup_{
\veta\not=0 } \frac{1}{
|\vxi|^2+|\veta|^2}\sum_{j=1}^N \widehat G_k(\vxi,\eta_j)\Gamma_{N-1}(\veta^j).
\end{equation}
In Proposition \ref{estim} we show that we can bound \eqref{supeta} in terms of ${\cal 
D}_1\left(\widehat G_k(\vxi,\cdot),|\vxi|\right)$ and of 
$|\partial_\eta^pG_k(\vxi,\eta)|$ for $p\leq4$, (see \eqref{C4}and \eqref{DD} 
below). Observe that ${\cal D}_1\left(\widehat G_k(\vxi,\cdot),|\vxi|\right)$ 
refers to the situation where there is only one particle in the reservoir $R$, 
and thus, the supremum is over $\eta\in\mathbb{R}$ instead of $\veta\in\mathbb{R}^N$.

Proposition \ref{Hi} then shows that $|\partial^4_\eta G_k(\vxi,\eta)|$ can be 
bounded in terms of the fourth moment $E_4$ of the initial distribution, (see 
\eqref{fourth}). We thus get a bound for $d_2(Q_I[l_k\Gamma_N],Q_B[l_k\Gamma_N])$ 
in terms of $d_2(Q_I[l_k\Gamma_1],Q_B[l_k\Gamma_1])$ and $E_4$. Together with 
\eqref{1to0} below, this will give us the desired estimate 
on $d_2(Q_I[l_k\Gamma_N],Q_B[l_k\Gamma_N])$ in terms of $d_2(l_k,\G_M)$. The 
conclusion of the proof of Theorem \ref{thd2} will then be 
very similar to the final steps of the proof of Theorem \ref{thL2}.

\begin{prop}\label{estim}
Let $H(\eta)$ be a bounded $C^4$ function of $\eta$. Assume that 
\[
 H(0)=0 \qquad\qquad H(\eta)=H(-\eta)
\]
and
\begin{equation}\label{C4}
C_4=\|H(\cdot)\|_{C^4}:=\max_{p\leq4}\sup_{\eta}\left|\frac{d^p}{d\eta^p} 
H (\eta)\right|<\infty.
\end{equation}
Calling
\begin{equation}\label{DD}
 {\cal D}_N(H,a)=\sup_{\veta\not=0}\frac{1}{
a^2+|\veta|^2}\left|\sum_{j=1}^N H(\eta_j)\Gamma_{N-1}(\veta^j)\right|
\end{equation}
we have
\begin{equation}\label{LL}
 {\cal D}_N(H,a)\leq \left[(8 C_4+{\cal D}_1(H,a)){\cal 
D}_1(H,a)\right]^{\frac12}
\end{equation}

\end{prop}

One may hope that ${\cal D}_N(H,a)\leq K {\cal D}_1(H,a)$ be true for some $K$ independent of $N$. 
We will show in Appendix \ref{P5} that no such $K$ exists.
Observe that $\mathcal D_N(H,a)$ is of order 1 uniformly in 
$N$ since we have

\begin{equation}\label{eq:crude}
 {\cal D}_N(H,a)\leq \sup_{\veta\not=0}\frac{\sum_{j=1}^N \left|
H(\eta_j)\right|}{\sum_{j=1}^N\eta_j^2}\leq
\sup_{\eta\not=0}\frac{|H(\eta)|}{\eta^2}={\cal D}_1(H,0).
\end{equation}
We were not able to use \eqref{eq:crude} directly. Indeed \eqref{eq:crude} and 
\eqref{GG} give 
\[
\frac1M d_2(Q_I[l_k\Gamma_N],Q_B[l_k\Gamma_N])
=\frac\mu N\sup_{\vxi,\eta\not=0 } \frac{1}{|\veta|^2}\sum_{j=1}^N \widehat 
G_k(\vxi,\eta)
\]
and it is not clear how to relate the right side of the above equation 
to $d_2(l_k,\G_M)$.

We can try to improve the above estimate observing that 
\begin{equation}\label{basic}
|H(\eta)|\leq {\cal D}_1(H,0)\eta^2
\end{equation}
so that
\[
\frac{\sum_{j=1}^N 
\left|H(\eta_j)\right|\Gamma_N(\veta^j)}{a^2+|\veta|^2}\leq 
{\cal D}_1(H,0)\Gamma_N(\veta)\frac{\sum_{j=1}^N
\eta_j^2e^{\pi \eta_j^2}}{a^2+|\veta|^2}.
\]
Since $xe^{\pi x}$ is an increasing function for $x>0$ we have that
\[
 \sum_{j=1}^N\eta_j^2e^{\pi \eta_j^2}\leq |\veta|^2 e^{\pi|\veta|^2}
\]
that is, the supremum of $\sum_{j=1}^N\eta_j^2e^{\pi \eta_j^2}$ on the set 
$|\veta|=N$ is reached when $\eta_1=N$ and $\veta^1=0$. This observation 
will be usefull in the following. Thus we get

\begin{equation}\label{eq:crude1}
\frac{\sum_{j=1}^N 
\left|H(\eta_j)\right|\Gamma_N(\veta^j)}{a^2+|\veta|^2}\leq {\cal 
D}_1(H,0)\frac{|\veta|^2}{a^2+|\veta|^2}.
\end{equation}
Alas, this is not yet enough since after taking the supremum on $\veta$ we are 
back to \eqref{eq:crude}. Observe though that, if $\bar \eta$ is such that 
$|H(\bar \eta)|=\sup_\eta |H(\eta)|$, then
\[
 \sup_{\veta\not=0}\frac{1}{
a^2+|\veta|^2}\left|\sum_{j=1}^N 
H(\eta_j)\Gamma_{N-1}(\veta^j)\right|=\sup_{\veta\not=0, 
|\eta_i|\leq\bar \eta}\frac{1}{
a^2+|\veta|^2}\left|\sum_{j=1}^N H(\eta_j)\Gamma_{N-1}(\veta^j)\right|
\]
that is, we can limit the seprumum in \eqref{DD} to the region where 
$\eta_i\leq 
\bar \eta$, for every $i$. But again we have no control on $\bar\eta$.
In the first part of the proof of Proposition \ref{estim} we will use 
an improved version of the above argument to show that ${\cal D}_N(H,a)$ can be 
bounded in terms of ${\cal D}_1(H,0)/(1+a^2)$.  

While it is obvious that ${\cal D}_1(H,a)\leq {\cal D}_1(H,0)$, the inverse 
inequality is generically far from true. In the second part of 
the proof, we find a lower bound on ${\cal D}_1(H,a)$ in 
terms of ${\cal D}_1(H,0)$ under the hypothesis that the fourth derivative of 
$H(\eta)$ is bounded. Observe indeed that, if $H(\eta)$ is of the form 
$H(\eta)=\frac{H''(0)}{2}\eta^2 - C \eta^4$ for 
some $C$, at least near $\eta=0$, then ${\cal D}_1(H, a) \geq \frac{ 
H''(0)^2}{ 2 a^2 C + H''(0) }$.  In Lemma \ref{firststep} we will show 
that a similar estimate holds for a generic $H$ once we replace $H''(0)$ by 
${\cal D}_1(H,0)$.

From these, Proposition \ref{estim} will easily follow.

\begin{proof}[Proof of Proposition \ref{estim}]

From \eqref{DD} it follows that
\[
 |H(\eta)|\leq {\cal D}_1(H,a)(\eta^2+a^2).
\]
Defining
\[
 \widetilde{H}(\eta,a)=\min\{{ \cal D}_1(H,0)\eta^2, {\cal D}_1(H,a) 
                       (a^2+\eta^2)\}  
                       =\begin{cases}
                       {\cal D}_1(H,0)\eta^2& {\rm if}\  \eta^2\leq\eta^2_0(a)\crcr
                       \mathcal D_1(H,a)(a^2+\eta^2) &{\rm if}\ \eta^2\geq\eta^2_0(a)
                      \end{cases} \ ,
\]
where 
\begin{equation}\label{eta0}
\eta^2_0(a)=\frac{\mathcal D_1(H,a)a^2}{{\cal D}_1(H,0)-\mathcal D_1(H,a)}
\end{equation}

\noindent is chosen to make $\widetilde{H}$ continuous.
We get $H(\eta)\leq \widetilde{H}(\eta,a)$ and thus $\mathcal D_N(H,a)\leq 
\mathcal D_N(\widetilde H,a)$. The following Lemma contains our main 
improvement of \eqref{eq:crude} and \eqref{eq:crude1}.

\begin{lem}\label{secondstep} Under the hypotheses of Proposition \ref{estim} we have
 \begin{equation}\label{eq:keta}
 {\cal D}_N(\widetilde H, a)={\cal D}_1(H,0)\sup_{k\leq N,|\eta|\leq\eta_0(a)} 
\frac{k\eta_0(a)^2e^{-\pi((k-1)\eta_0(a)^2+\eta^2)}+\eta^2e^{
-\pi k\eta_0(a)^2}}{a^2+k\eta_0(a)^2+\eta^2}
\end{equation}
that is, the supremum in \eqref{DD} for $\widetilde H$ is attained for 
$\veta$ of the form $\veta=(\eta_0(a),\ldots,\eta_0(a),\eta,0,\ldots,0)$ for some $\eta$ with 
$|\eta|\leq\eta_0(a)$. 

\end{lem}

\begin{proof} Let
\[
 \widetilde{\cal H}_N(a,\veta)=\frac{\sum_{i=1}^N\widetilde 
 H(\eta_i)\Gamma_{N-1}(\veta^i)}{a^2+|\veta|^2}
\]
and suppose $ \veta$ has $\vert \eta_i \vert > \eta_0(a)$ 
for some $i$. By differentiating we get 
\[
\partial_{\eta_i} \widetilde{\cal 
H}_N(a,\veta)=\partial_{\eta_i}\left(\widetilde 
H(a,\eta_i)e^{\pi\eta_i^2}\right)\frac{\G_N(\veta)
}{a^2+\veta^2}-2\eta_i\left(\pi+\frac{1}{a^2+
\veta^2}\right)\widetilde{\cal H}_N(a,\veta)
\]
where we used
\[
 \partial_{\eta}\left(\widetilde 
H(a,\eta)e^{\pi\eta^2}\right)=2\eta\left(\pi\widetilde H(a,\eta)+{\cal 
D}_1(H, a)\right)e^{\pi\eta^2}
\]
whenever $\eta\geq\eta_0(a)$. Because
\[
\frac{\widetilde H(a,\eta_i)\G_{N-1}(\veta^i)
}{a^2+\veta^2}\leq \widetilde{\cal H}_N(a,\veta)\qquad\hbox{and} 
\qquad\frac{{\cal D}_1(\widetilde H, a)\G_{N-1}(\veta^i)
}{a^2+\veta^2}\leq \frac{{\cal 
D}_N(\widetilde H, a)}{a^2+\veta^2}\ ,
\]
with equality holding only if $\veta^i=0$, we have
\[
 \partial_{\eta_i} \widetilde{\cal H}_N(a,\veta) <0.
\]
 This implies that 
\[
\sup_{\veta\not=0}\widetilde{\cal 
H}_N(a,\veta)=\sup_{\veta\not=0,|\eta_i|\leq\eta_0}\widetilde{\cal 
H}_N(a,\veta)\, .
\]

Now we show that there can be at most $1$ coordinate $i$ such that $0< \vert 
\eta_i \vert < \eta_0(a)$. Consider now
\[
 L(x,y):=x^2e^{\pi x^2}+y^2e^{\pi y^2}
\]
and observe that $L(r\cos\theta,r\sin\theta)$ is maximal for 
$\theta=n\frac\pi2$ and minimal for $\theta=\frac\pi4+n\frac\pi2$. Moreover, it 
is strictly increasing for $\frac\pi4+n\frac\pi2<\theta<(n+1)\frac\pi2$ and strictly 
decreasing for $n\frac{\pi}{2}<\theta<n\frac{\pi}{2}+ \frac{\pi}{4}$.
For $|\eta_i|\leq \eta_0(a)$ we have
\[
 \widetilde{\cal H}_N(a,\veta)=\frac 
{{\cal D}_1(H,a)L(\eta_1,\eta_2)\Gamma_{N-2}(\eta_3\ldots,\eta_N)+\sum_{i=3}^N 
\widetilde H(a,\eta_i)\Gamma_{N-1}(\veta^i)}{ a^2+|\veta^2|}\, ,
\]
so that there can be no maximum for $\widetilde{\cal H}_N(a,\veta)$ for which 
both $0<\eta_1<\eta_0(a)$ and $0<\eta_2<\eta_0(a)$. Repeating this argument 
for each pair $\eta_i,\eta_j$ with $1\leq i,j\leq N$ we get that for all but possibly one $i$, 
we must have $\eta_i=0$ or $\eta_i=\eta_0(a)$.
\end{proof}

To complete the proof of the first part of Proposition \ref{estim} we will 
simplify the right hand side of equation \ref{eq:keta}. Observe first that
\[
 \frac{k\eta_0(a)^2e^{-\pi((k-1)\eta_0^2(a)+\eta^2)}+\eta^2e^{
-\pi k\eta_0(a)^2}}{a^2+k\eta_0(a)^2+\eta^2}\leq 
\max\left\{\frac{\eta_0^2(a)}{\frac{a^2}2+\eta_0(a)^2}, 
\frac{(k-1)\eta_0(a)^2e^{-\pi((k-1)\eta_0^2(a)+\eta^2)}+\eta^2e^{
-\pi k\eta_0(a)^2}}{\frac{a^2}2+(k-1)\eta_0(a)^2+\eta^2}\right\}.
\]
From \eqref{eta0} we have
\[
 \frac{\eta_0^2(a)}{\frac{a^2}2+\eta_0(a)^2}\leq 2\frac{\mathcal D_1(H, 
a)}{\mathcal D_1(H,0)}
\]
while
\begin{align}
\sup_{k\leq N,|\eta|\leq\eta_0(a)} 
\frac{(k-1)\eta_0(a)^2e^{-\pi((k-1)\eta_0^2(a)+\eta^2)}+\eta^2e^{
-\pi k\eta_0(a)^2}}{\frac{a^2}2+(k-1)\eta_0(a)^2+\eta^2}&\leq \crcr
\sup_{k\leq N,|\eta|\leq\eta_0(a)} 
\frac{((k-1)\eta_0(a)^2+\eta^2)e^{-\pi((k-1)\eta_0^2(a)+\eta^2)}}{\frac{a^2}
2+(k-1)\eta_0(a)^2+\eta^2}&\leq 2\sup_{y>0} \frac{ye^{-\pi 
y}}{\frac{a^2}2+y}
\end{align}
Clearly we have
\[
 \frac{ye^{-\pi y}}{\frac{a^2}2+y}\leq \frac{y}{(\frac{a^2}2+y)(1+\pi y)}\leq
\frac1{\frac{\pi a^2}2+1}
\]
so that 
\begin{equation}\label{N}
 \mathcal D_N(H,a)\leq \max\left\{{\cal D}_1(H, a),2\frac{{\cal D}_1(H, 
0)}{1+\frac{\pi}{2}a^2}\right\}\, .
\end{equation}

\noindent This concludes the first part of the proof. We start the second part with a 
couple of simple observations.
%
%
%
%

From the hypotheses of Proposition \ref{estim}, it follows that

\begin{equation} \label{approx}
\frac{|H''(0)|\eta^2 }{2} - \frac{C_4 \eta^4}{4!} \le |H(\eta)| 
\le 
\frac{|H''(0)|\eta^2 }{2} + \frac{C_4 \eta^4}{4!}  \ .
\end{equation}

  Let now $M=\sup_\eta |H(\eta)|$ and observe that there 
exists a finite $\tilde \eta$ such that $|H(\tilde\eta)|>M/2$. Moreover 
$\tilde\eta\not=0$ since $H(0)=0$. Thus $\mathcal{D}_1(H,0)\geq 
M/(2\tilde\eta^2)$ while
\[
 \frac{|H(\eta)|}{\eta^2}< \frac{M}{2\tilde\eta^2}\qquad\hbox{if}\qquad 
\eta^2> 2\tilde\eta^2
\]
Thus there exists $\eta_m$ such that $\eta_m^2\leq \tilde\eta^2$ and 
$|H(\eta_m)|=\mathcal{D}_1(H,0)\eta_m^2$. We also know from \eqref{basic} that 

\[ |H''(0)|\leq 2\mathcal{D}_1(H,0),\]

\noindent  with equality if and only if $\eta_m^2 = 0$.


\begin{lem}\label{firststep} Under the hypotheses of Proposition \ref{estim} we 
have
\[
\mathcal D_1(H,a) \ge \frac{\mathcal D_1(H,0)^2}{\frac{3}{2}C_4a^2 +4\mathcal 
D_1(H,0)}
\]
\end{lem}

\begin{proof} 
From \eqref{approx} it follows that
\[
\frac{|H(a,\eta)|}{a^2+\eta^2} \ge \frac{\frac{|H''(0)|\eta^2 }{2} - 
\frac{C_4 \eta^4}{4!}}{a^2+\eta^2}
\]
and, choosing $\eta^2$ to be $\frac{6 \vert H''(0)\vert }{C_4}$, we get 
that

\begin{equation}\label{up}
\sup_ \eta \frac{|H(a,\eta)|}{a^2+\eta^2} \ge  \frac{ \vert 
H''(0)\vert^2}{ 4  \vert H''(0)\vert+ \frac{3}{2}C_4
a^2}.
\end{equation}
%
%

Since, there is no positive lower bound for $\vert H''(0)\vert$, we complement this inequality using the second inequality in\eqref{approx}. We find that for all $\eta$
\[
|H(\eta)|-\mathcal{D}_1(H,0)\eta^2\le 
\frac{(|H''(0)|-2\mathcal{D}_1(H,0))\eta^2 }{2} + \frac{C_4 
\eta^4}{4!} 
\]
Since $|H''(0)|-2\mathcal{D}_1(H,0)\leq 0$ we get
\[
 \eta_m^2\geq \frac{12(2\mathcal{D}_1(H,0)-|H''(0)|)}{C_4}\ .
\]
%

This implies that
\begin{equation}\label{down}
\sup_ \eta \frac{|H(\eta)|}{a^2+\eta^2} \geq  
\frac{|H(\eta_m)|}{a^2+\eta_m^2}\geq \liminf_{\epsilon\rightarrow 0}  
 \frac{|H(\eta_m)|}{\eta_m^2 +\epsilon} \frac{\eta_m^2}{a^2+\eta_m^2} \ge  
\frac{12\mathcal{D}_1(H,0)(2\mathcal{D}_1(H,0)-|H''(0)|)}{C_4a^2 
+12(2\mathcal{D}_1(H,0)-|H''(0)|)}\ .
\end{equation}

Observe now that the right hand side of \eqref{up} is an increasing function of 
$|H''(0)|$ while the right hand side of \eqref{down} is decreasing.
Thus, we have

\[
 \mathcal{D}_1(H,a)\geq \min_{0\leq h \leq 2\mathcal{D}_1(H,0)} \max \left\{  \frac{ h^2}{ 4  h+ \frac{3}{2}C_4 a^2}\, ,\,
\frac{12\mathcal{D}_1(H,0)(2\mathcal{D}_1(H,0)-h)}{C_4a^2 
+12(2\mathcal{D}_1(H,0)-h)}\right\}
\]

Moreover
\begin{align*}
 \frac{12\mathcal{D}_1(H,0)(2\mathcal{D}_1(H,0)-h)}{C_4a^2 +12 
(2\mathcal{D}_1(H,0)-h)} \geq \frac{12 \mathcal{D}_1(H,0)^2}{ 12 
\mathcal{D}_1(H,0)^2+ C_4 a^2}&\qquad\hbox{for}\qquad &h\leq 
\mathcal{D}_1(H,0)\\
\frac{ 
h^2}{ 4  h+ \frac{3}{2}C_4 a^2} \geq  \frac{\mathcal 
D_1(H,0)^2}{\frac{3}{2}C_4a^2 +4\mathcal D_1(H,0)}&\qquad\hbox{for}\qquad &
h\geq  \mathcal{D}_1(H,0).
\end{align*}

The above, together with the observation 
\[
\frac{\mathcal D_1(H,0)^2}{\frac{3}{2}C_4a^2 +4\mathcal 
D_1(H,0)} \leq \frac{12 \mathcal{D}_1(H,0)^2}{ 12 \mathcal{D}_1(H,0)^2+ C_4 
a^2}\ 
\]
\noindent concludes the proof.
\end{proof}

Observe finally that from $2|H(\eta)|/\eta^2 \leq \sup_{\eta}|H''(\eta)|$ it follows
that $2\mathcal D_1(H,0)\leq\sup_{\eta}|H''(\eta)|\leq C_4$. Thus we can 
write
\begin{equation}\label{1}
 \mathcal D_1(H,a) \ge \frac{2\mathcal D_1(H,0)^2}{C_4}\frac{1}{3a^2 
+4} \ .
\end{equation}

Putting together \eqref{N} and  \eqref{1} establishes the claim of Proposition \ref{estim}.

\end{proof}

To apply Proposition \ref{estim} to \eqref{supeta}, we need to estimate $\Vert 
\widehat G_k(\vxi,\cdot)\Vert_{C^4}$, where $\widehat G_k(\vxi,\eta)$ is 
defined in \eqref{bGG}. Observe that for $p \le 4$ we have by Jensen's inequality
\begin{align*}
\left| \partial_{\eta_j}^p \widehat{R}^I_{i,j}[\widehat 
l_k\Gamma_N](\vxi,\veta)\right|\leq& (2\pi)^4 \int |w_j|^p
R^I_{i,j}[l_k\Gamma_N](\V,\W)d\V d\W \leq  (2\pi)^4 \left(\int |w_j|^4
R^I_{i,j}[l_k\Gamma_N](\V,\W)d\V d\W\right)^{\frac p4}=\\
=&(2\pi)^4\left(\int (w_j^2+v_i^2)^{2}l_k(\V)\Gamma_N(\W)d\V d\W 
\right)^{\frac p4}=(2\pi)^4 
\left(E_{4,k}+2\frac{E_{2,k}}{\sqrt{2\pi}}+\frac3{2\pi}\right)^{\frac p4}\leq 
32\pi^4 ( E_{4,k} + 1)
\end{align*}
where 
\[   
 E_{n,k}=\int v_i^nl_k(\V)d\V=\int v_i^n \left(Q_S+Q_B+\frac{\lambda_R 
N}{2}I\right)^k [l_0](\V)d\V\ .
\]
Using \eqref{bGG} we get
\begin{equation} \label{cee}
\Vert \widehat G_k(\vxi,\cdot)\Vert_{C^4}\leq 32 \pi^4 
\left(E_{4,k}+1\right).
\end{equation}

To estimate $E_{4,k}$ we need to study the action of $Q_S$ and $Q_B^*$ on 
$v_i^4$, where $Q_B^*$ is the adjoint of $Q_B$. This is done in the following 
Lemma.

\begin{prop}\label{Hi} Given a symmetric distribution $l_0$ on 
$\mathbb{R}^M$ 
such that 
\[
 \int v_i^4l_0(\V)d\V=E_4<\infty
\]
we have
\[
 E_{4,k}=\int v_i^4l_k(\V)d\V \leq 2 (E_4+1)
\]
where $l_k=\Lambda^{-k}\left(Q_S+Q_B+\frac{\lambda_R N}{2}I\right)^k l_0$.

\end{prop}

\begin{proof}
First we observe that, due to symmetry,
\[
E_{4,k}=\int \frac{1}{M}\sum_{i=1}^M v_i^4l_k(\V)d\V.
\]
Calling
\[
 \overline Q_S :=\frac{1}{\binom{M}{2}}\sum_{i<j} 
R^S_{i,j}=\frac{2}{\lambda_S M}Q_S\qquad,\qquad \overline Q_B := \frac{1}{M} 
\sum_{i=1}^M B_i=\frac{1}{\mu M} Q_B
\]
we have that
\[
 \int v_i^4 \overline Q_S[l](\V)d\V=\int \overline Q_S[v_i^4] l(\V)d\V\qquad
 \int v_i^4 \overline Q_B[l](\V)d\V=\int \overline Q_T[v_i^4] l(\V)d\
\]
where
\[
 \overline Q_T := \frac{1}{M}\sum_{i=1}^M T_i
\]
with $T_i$ defined in \eqref{Ti}. It is easy to see that $\overline Q_S$ and 
$\overline Q_T$ leave the space $V$ of even polynomials of degree at 
most 4 in the $v_i$ invariant. Calling $H_n(v)$ the monic Hermite polynomial of 
degree
$n$ 
(with weight $\Gamma_1(v)=e^{-\pi v^2}$), a natural basis in $V$ is given by
\[
 {\cal H}_4(\V)=\frac{1}{M}\sum_{i=1}^M H_4(v_i)\qquad  {\cal 
H}_3(\V)=\frac{2}{M(M-1)}\sum_{i<j} H_2(v_i)H_2(v_j)\qquad  {\cal 
H}_2(\V)=\frac{1}{M}\sum_{i=1}^M H_2(v_i)\qquad {\cal H}_0(\V)=1
\]
and we have
\begin{eqnarray*}
 \frac{1}{M}\sum_{i=1}^M v_i^4 & =& a_4 {\cal H}_4(\V)+a_3 {\cal 
H}_3(\V)+a_2{\cal H}_2(\V)+a_0{\cal H}_0(\V)\\ 
\end{eqnarray*}
where $\vec a=(a_4,a_3,.a_2,a_0)=(1,0,\frac{3}{\pi},\frac{3}{4\pi^2})$ and 
$|\vec a|\leq \sqrt{2}$. From \cite{BLV} we know that the action of $\overline 
Q_S$ and $\overline Q_T$ on $V$ with the basis ${\cal H}_i$ is given by two 
positive definite matrices $L_S$ and $L_T$ with spectral (and thus $L^2$) norm 1. Thus also the action of 
$\Lambda^{-k}\left(Q_S+Q_T+\frac{\lambda_R N}{2}I\right)$ is given by a 
positive definite matrix $L$ with norm 1. Thus we get
\[
 \Lambda^{-k}\left(Q_S+Q_T+\frac{\lambda_R 
N}{2}I\right)\left(\frac{1}{M}\sum_{i=1}^M v_i^4\right)=a_{4,k} 
{\cal H}_4(\V)+a_{3,k} {\cal H}_3(\V)+a_{2,k}{\cal H}_2(\V)+a_{0,k}{\cal H}_0(\V)
\]
where $\vec a_k=L^k\vec a$. Clearly we have $\vert\vec a_k\vert \leq 
\vert\vec a\vert \leq \sqrt{2}$. We integrate both sides against 
$l_0(\V)$ to obtain 

\begin{eqnarray*}
E_{4,k} & = & a_{4,k} \left(E_4 - \frac{3}{\pi} E_2 + \frac{3}{4\pi^2}\right) + 
a_{3,k} \left(E_3 -\frac{1}{\pi} E_2 + \frac{1}{4\pi^2}\right)+ a_{2,k} 
\left(E_2 - \frac{1}{2\pi}\right) + a_{0,k}\\
\end{eqnarray*}
where 
\[
 E_2=\int v_i^2 l_0(\V)d\V\leq \frac{1}{2}\left( 1+ E_4 \right)\qquad\qquad E_3=\int v_i^2 v_j^2 l_0(\V)d\V\leq E_4.
\]
After some rearranging and neglecting terms with negative coefficients, we 
obtain
\begin{eqnarray*}
 E_{4,k} & \leq & E_4 \left(  \left(1 - \frac{3}{2\pi}\right)a_{4,k} +  
\left(1 - \frac{1}{2\pi}\right)a_{3,k} + \frac{1}{2}a_{2,k}   \right) 
+ \left( a_{0.k}+  \left(\frac{1}{2} - \frac{1}{2\pi}\right)a_{2,k}\right) \\
& \leq & |\vec a| \left(E_4 \sqrt{ \left(1- 
\frac{3}{2\pi}\right)^2+\left(1-\frac{1}{2\pi}\right)^2+ \frac{1}{4}} + \sqrt{ 
1 + \left(\frac12 - \frac{1}{2\pi}\right)^2}\right)\\
\end{eqnarray*}

\noindent proving the result. Here we applied Cauchy-Schwarz inequalities 
twice in the last step.
\end{proof}

It thus follows from \eqref{cee} that
\begin{equation} \label{ceefour}
 \|G_k(\vxi,\cdot)\|_{C^4}\leq 96 \pi^4 (E_4+1):=2F_4.
\end{equation}
Applying Proposition \ref{Hi} and Proposition \ref{estim} to \eqref{GG}, \eqref{supeta} and using \eqref{ceefour} we get that
\begin{equation}\label{wow}
 d_2(Q_I[l_k\Gamma_N],Q_B[l_k\Gamma_N])\leq \frac{\mu 
M}{N}\sqrt{\bigl( 2K F_4+(\mu M)^{-1}d_2(Q_B[l_k\Gamma_1],Q_I[l_k\Gamma_1])\bigr)(\mu M)^{-1}d_2(Q_B[
l_k\Gamma_1 ] , Q_I [ l_k\Gamma_1])} \ ,
\end{equation}
where $K$ is defined in Theorem \ref{thd2}.
It is easy to see that
\begin{equation}\label{1to0}
\frac 1 M d_2(Q_I[l_k\Gamma_1],Q_B[l_k\Gamma_1])\leq 
d_2(M^{-1}Q_I[l_k\Gamma_1],\mu \G_{M+1})+d_2(M^{-1}Q_B[l_k\Gamma_1],\mu 
\G_{M+1})\leq 
2\mu d_2(l_k,\G_{M}).
\end{equation}
Combining \eqref{wow} and \eqref{1to0} gives  
\begin{equation}\label{middled2}
 d_2(Q_I[l_k\Gamma_N],Q_B[l_k\Gamma_N])\leq 2 \frac{\mu 
M}{N}\sqrt{(8F_4+d_2(l_k,\G_{M}))d_2(l_k,\Gamma_M)}
\end{equation}

We can now conclude our proof. Indeed, going back to eq\eqref{expad22}, we can 
write
\begin{align*}
d_2\left(e^{\L t}f_0,e^{\widetilde \L t}f_0\right)&\leq 2 \frac{\mu 
M}{N}e^{-\Lambda t}\sum_{n=1}^\infty 
\frac{t^n}{n!}\sum_{k=0}^{n-1}\Lambda^{n-1}\sqrt{(8 F_4+d_2(l_k,\G_{M}
))d_2(l_k , \Gamma_M)}\crcr
&\leq 2\frac{\mu M}{N}e^{-\Lambda 
t}\sum_{n=1}^\infty 
\frac{t^n\Lambda^{n-1}}{n!}\sum_{k=0}^{n-1}\left(1-\frac{\mu}{2\Lambda}\right)^{
\frac k2}\sqrt{(8 F_4+d_2(l_0,\G_{M}))d_2(l_0 , \Gamma_M)}\\
&= 8\frac{M}{N}\left(1-e^{-\frac\mu4 
t}\right)\sqrt{(8F_4+d_2(l_0,\G_{M}))d_2(l_0 , \Gamma_M)}
\end{align*}
where we have used \eqref{befored2} in Lemma \ref{RB} together 
with $\left(1-\frac{\mu}{2\Lambda}\right)^{
\frac 12}\leq 1-\frac{\mu}{4\Lambda}$.

\section{Conclusions and Outlooks}\label{conc}

We have shown that a {\it small} system out of equilibrium interacting with a 
{\it large} system initially in equilibrium (the reservoir) can be
well approximated in certain norms by a the same small system interacting with
a thermostat. This approximation moreover is uniform in time. Our proof is
not based on a projection or conditioning method. Indeed, it is hard to see how
one can apply such an argument to the $d_2$ metric. In particular we obtain
that also the reservoir remains uniformly close to the equilibrium state.

We can also think of our system as describing a local perturbation in a large
system initially in equilibrium at a given temperature. In this spirit we see
our results as an initial attempt to understand the return to equilibrium from
an initial state that is locally close to equilibrium. We hope to come back on
this problem on forthcoming research.

In the case of the $L^2$ norm introduced in section \ref{L2}, the derivation of 
the above approximation is rather direct. We believe that this is at least in
part due to the fact that the generators $\L$ (see \eqref{genr}) and $\overline
\L$ (see \eqref{gent}) both have a spectral gap uniform in $N$. This implies
that both systems approach exponentially fast to their respective steady states
$f_{\rm \infty}$ and $\tilde f_{\rm \infty}$, \eqref{ssr} and \eqref{sst}.
Notwithstanding this, such a norm behaves poorly with the size of the system and
it excludes altogether perfectly reasonable initial states.

Partly for this reason we have studied the $d_2$ metric defined in \eqref{dd2}.
Such a metric is well defined for all reasonable initial states and behaves much
better as a function of the size of the system. The control of this norm is
harder. The main ingredient is contained in Proposition \ref{estim} in section
\ref{pd2ev}. It requires an extra fourth moment assumption on the initial state and
some substantial analysis of an associated functional inequality.

It is not hard to show that $ e^{\widetilde{\cal L} t}f_0$ approaches $\tilde
f_\infty$ exponentially fast in the $d_2$ metric (see \cite{Evans2016, CLM}). 
On the other
hand, it is an open question whether $e^{{\cal L} t}f_0$ approaches $f_\infty$
exponentially fast in the $d_2$ metric at a rate uniform in $N$. Our
result is not enough to give an answer but it makes such a question rather natural.

Finally in \cite{CLM}, the authors consider a system interacting with more then 
one thermostat. They start at the level of the Boltzmann equation but it would 
be interesting to see in which sense one can approximate such a system with a
system interacting with several large but finite reservoirs at different 
temperatures. Observe that in such a case, if the reservoirs are kept finite, 
they will reach a steady state in which they all have the same temperature (or 
better, average kinetic energy). This will create a more complex and 
interesting interplay between the large $N$ and large $t$ limit, with more than 
one  time scale involved. 

\appendix

\section{Estimates on the Steady States}

In this Appendix we derive \eqref{L2eq} and \eqref{d2eq}.

\subsection{Derivation of \eqref{L2eq}}\label{sL2eq}

Because $h_{\infty}$ depends only on $r=\sqrt{|\V|^2+|\W|^2}$ we can set
\[
 H(r)=h_{\infty}(\V,\W)
\]
Moreover, setting 
\[
w_j = \tilde w_j \sqrt{r^2 - |\V|^2} 
\]
we get $r^2 - |\W|^2 =(r^2 - |\V|^2) (1 -|\tilde \W|^2)$ and
\[
H(r) =
\frac{2}{|\Sp^{M+N-1}|r^{M+N-1}} \int_{|\V|^2 \le r^2} h_0(\V)   r  \left(r^2 - 
|\V|^2\right)^{\frac{N-2}2} d \V
\int_{ \sum_{i \le N-1} w_i^2 \le 1} \frac{1}  {\sqrt{
1 - \sum_{j=1}^{N-1} \tilde w_j^2}} 
d \tilde w_1 \cdots d \tilde w_{N-1}
\]
so that we have
\[
H(r) = \frac{ |\Sp^{N-1}|}{|\Sp^{M+N-1}|r^M} \int_{\R^M} h_0( \V) 
\left(1- 
\frac{|\V|^2}{r^2} \right)_+^{(N-2)/2} 
 d \V 
\]
where $(x)_+=x$ if $x\geq 0$ and $(x)_+=0$ otherwise. Because $\int 
\G_N(\V) h_0(\V) d\V =1$  and
\[
 \frac{ |\Sp^{N-1}|}{|\Sp^{M+N-1}|r^M} \int_{\R^M}\left(1- 
\frac{|\V|^2}{r^2} \right)_+^{(N-2)/2} d \V=1
\]
we may write
\begin{align*}
H(r)-1 =& \int_{\R^M}  \left[ \frac{ |\Sp^{N-1}|}{|\Sp^{M+N-1}|r^M}  
\left(1- 
\frac{ 
|\V|^2}{r^2} \right)_+^{(N-2)/2}  - \G_N(\V)\right] (h_0(\V) -1)d \V\\
=&\int_{\R^M}  \left[ \frac{ |\Sp^{N-1}|}{|\Sp^{M+N-1}|r^M}  \left(1- 
\frac{ 
|\V|^2}{r^2} \right)_+^{(N-2)/2}e^{\pi|\V|^2/2}  - e^{-\pi|\V|^2/2}\right] 
e^{-\pi|\V|^2/2} (h_0(\V)-1) d \V
\end{align*}
and using Cauchy-Schwarz's inequality we find that
\[
|H(r) -1|^2 \le \int_{\R^M} \G_N(\V)( h_0(\V)-1)^2 d \V
 \int_{\R^M} \left[ \frac{ |\Sp^{N-1}|}{|\Sp^{M+N-1}|r^M}  \left(1- 
 \frac{ |\V|^2}{r^2} \right)_+^{(N-2)/2}e^{\pi|\V|^2/2}  - 
e^{-\pi|\V|^2/2}\right]^2 d\V \ .
\]
Thus, we get
\[
 \|h_{\infty}-1\|^2=|\Sp^{M+N-1}|\int r^{M+N-1} e^{-\pi r^2} |H(r)-1|dr\leq 
C\|h\|^2_2
\]
where
\[
C=|\Sp^{M+N-1}| \int_0^\infty  d r r^{M+N-1} e^{-\pi r^2} \int_{\R^M}  \left[ 
\frac{ |\Sp^{N-1}|}{|\Sp^{M+N-1}|r^M}  \left(1- \frac{ |\V|^2}{r^2} 
\right)_+^{(N-2)/2}e^{\pi|\V|^2/2}  - e^{-\pi|\vec v|^2/2}\right]^2 d \V 
\]
By expanding the square, we can write the above integral as a sum of three 
integrals that can be computed explicitly as
\begin{align}
\int_0^\infty  d r r^{M+N-1} e^{-\pi r^2} \int_{\R^M}  \frac{ 
|\Sp^{N-1}|^2}{|\Sp^{M+N-1}| r^{2M}}  \left(1- \frac{ |\V|^2}{r^2} 
\right)_+^{(N-2)}e^{\pi|\V|^2} d \V&= \frac{ \Gamma(\frac{M+N}{2}) 
}{\Gamma(\frac{N}{2}) \Gamma(\frac{M}{2})}   
\frac{\Gamma(\frac{N-2}{2}) \Gamma(\frac{M}{2})}{\Gamma(\frac{M+N-2}{2})} = 
\frac{M+N-2}{N-2} \ ,  \crcr
\int_0^\infty  d r r^{M+N-1} e^{-\pi r^2} \int_{\R^M}   \frac{ 
|\Sp^{N-1}|}{r^M}  \left(1- \frac{ |\V|^2}{r^2} \right)_+^{(N-2)/2} 
d\V&=1 \ ,\crcr
|\Sp^{M+N-1}| \int_0^\infty  d r r^{M+N-1} e^{-\pi r^2} \int_{\R^M}  
e^{-\pi|\V|^2} d \V&=1 \ .
\end{align}
We thus get 
\[
C=\frac{M}{N-2} \ .
\]

\subsection{Derivation of \eqref{d2eq}}\label{sd2eq}

Calling $r^2=|\vxi|^2+|\veta|^2$, we have
\begin{align*}
d_2(f_{\infty},\Gamma_{M+N}) =& \sup_{r\not= 0}\int_{\Sp^{M+N-1}(r)}
\frac{[ \widehat l_0(\vxi) - \Gamma_M(\vxi) 
]}{r^2}\Gamma_N (\veta)  d \sigma_r(\vxi,\veta)\crcr
\le& \left(\sup_{ r\not= 0}\int_{\Sp^{M+N-1}(r)}\frac{|\vxi|^2}{r^2} 
\Gamma_N (\veta)  d \sigma_r (\vxi,\veta)\right) d_2(l_0,\Gamma_M)
\end{align*}
Observe now that
\begin{align*}
\int_{\Sp^{M+N-1}(r)}\frac{|\vxi|^2}{r^2} \Gamma_N (\veta)  d \sigma_r
(\vxi,\veta)&=\int_{\Sp^{M+N-1}(1)}|\vxi|^2 
\gamma\left(r^2(1-|\vxi|^2)\right)d\sigma_1 (\vxi,\veta)\leq
\frac{|\Sp^{N-1}|}{|\Sp^{M+N-1}|}\int_{|\vxi|^2 \le 1}|\vxi|^2\left(1 -  
|\vxi|^2\right)^{\frac{N-2}{2}}d\vxi\leq\crcr
&\leq\frac{|\Sp^{N-1}||\Sp^{M-1}|}{|\Sp^{M+N-1}|}\int_0^1\rho^{M+1}\left(1 -  
\rho^2 \right)^{\frac{N-2}{2}} d \rho=\frac12\frac{|\Sp^{N-1}| 
|\Sp^{M-1}|}{|\Sp^{M+N-1}|} \int_0^1 s^{\frac M2} (1 -  s 
)^{\frac N2-1} ds=\crcr
=& \frac12\frac{ 2 \pi^{\frac M2} 2\pi^{\frac N2} 
\Gamma\left(\frac{M+N}2\right)}{\Gamma\left(\frac M2\right) 
\Gamma\left(\frac N2\right)2 \pi^{\frac{M+N}2}}\frac{\Gamma\left(\frac 
M2+1\right) \Gamma\left(\frac N2\right) 
}{ \Gamma\left(\frac{M+N}2+1\right) } = \frac{M }{ M+N }. 
\end{align*}

\section{Optimality of the estimate \eqref{middleL2}}\label{M}

In this appendix we show that there exists an initial state $u_0$ for which we have
\[
 \|(Q_I-Q_T)u_0\|_2\geq C \frac{M}{\sqrt{N}}\|u_0\|_2.
\]
thus saturating the bound in Lemma \ref{fl}.
We first observe that, by a similar analysis as Lemma \ref{fl}, we get
\[
\left\Vert \sum_{i=1}^M\left(\frac{1}{N} \sum_{j=1}^N R^I_{i,j} u - T_i u\right)
\right\Vert^2_2  = 
\frac{M}{N}\left(\langle T_1 u, u\rangle - \langle T_1 u,T_1 
u\rangle\right)
+\frac{M(M-1)}{N}\left(\langle R^I_{1,1} u, R^I_{2,1} u\rangle-
\langle T_{1} u, T_{2} u\rangle\right).
\]
We thus need symmetric initial states such that $\langle R^I_{1,1} u, 
R^I_{2,1} u\rangle- \langle T_{1} u, T_{2} u\rangle=O(1)$ in $M$ and $N$.
To this end we set
\[
 u_{M,P}(\V)=\sum_{p_1+p_2+\cdots +p_M=P}\prod_{i=1}^M\,\, H_{2p_i}(v_i)
\]
where $H_p(v)$ is the normalized Hermite polynomial of degree $p$ with weight 
$\gamma(v)=e^{-\pi v^2}$. We get
\[
 R^I_{1,1}u_{M,P}(\V)=\sum_{p_1+p_2\leq 
P}\widetilde H_{2p_1}(v_1,w_1)H_{2p_2}(v_2)u_{M-2,P-p_1-p_2}(\V^{1,2})\ .
\]
where $\widetilde H_{2p}(v,w)$ is the only radially symmetric Hermite 
polynomial of degree $2p$. It follows that
\[
\langle R^I_{1,1} u_{M,P}, R^I_{2,1} u_{M,P}\rangle- \langle T_{1} u_{M,P}, 
T_{2} u_{M,P}\rangle\geq
\left(\langle R^I_{1,1}\bar u,R^I_{2,1}\bar 
u\rangle-\langle T_{1}\bar u,T_{2}\bar 
u\rangle\right)\|u_{P-2,M-2}\|_2
\]
where $\bar u(v_1,v_2)=H_4(v_1)+H_2(v_1)H_2(v_2)+H_4(v_2)$. Observe now that
$\|u_{P,M}\|_2=\binom{M+P}{P-1}$ while $\langle R^I_{1,1}\bar u,R^I_{2,1}\bar 
u\rangle-\langle T_{1}\bar u,T_{2}\bar 
u\rangle=\frac{11}{8}$ so that
\[
\langle R^I_{1,1} u_{M,P}, R^I_{2,1} u_{M,P}\rangle- \langle T_{1} u_{M,P}, 
T_{2} u_{M,P}\rangle\geq
\frac{11}{8}\frac{(P-1)(P-2)(M+1)M}{(M+P)(M+P-1)(M+P-2)(M+P-3)}\|u_
{ M , P } \|_2.
\]
By choosing $P=M$ we get 
\[
 \langle R^I_{1,1} u_{M,M}, R^I_{2,1} u_{M,M}\rangle- \langle T_{1} u_{M,M}, 
T_{2} u_{M,M}\rangle\geq C \|u_{M,M}\|_2
\]
with $C=3/128$.

We can thus consider an initial state given by
\[
 h_0(\V)=1+a u_{M,M}(\V).
\]
Observe that $u_{M,M}$ is an even polynomial in all its variables with positive 
coefficients for the terms of maximal degree. Thus $\inf_{\mathbb{R}^n} 
u_{M,M}(\V)>-\infty$ and choosing $a$ small enough we get $h_0\geq 0$.

Going back to \eqref{fin} we can write
\begin{align*}
\|(e^{\L t}-e^{\overline \L t}) h_0\|_2&\geq \|h_0-1\|_2
\frac{M}{\sqrt{N}}e^{-\Lambda t}\left(Ct-\sum_{n=2}^\infty 
\frac{t^n}{n!}\Lambda^n\left[1-\left(1-\frac{\mu 
}{2\Lambda}\right)^n\right]\right)\\
&\geq\|h_0-1\|_2
\frac{M}{\sqrt{N}}t\left((C+1)e^{-\Lambda t}-1\right) 
\end{align*}
where we have used that $\left[1-\left(1-x\right)^n\right]\leq nx$. Thus for 
this particular $h_0$ our estimate is saturated at least for a time order 
$\Lambda^{-1}$. Since $\Lambda>(\lambda_S/2+\mu)M$ we cannot claim that for 
this example  $\|(e^{\L t}-e^{\overline \L t}) h_0\|_2$ actually grows to order
$M/\sqrt{N}$.

\section{Violation of  ${\cal D}_N(H,a)\leq K {\cal D}_1(H,a)$}\label{P5}
In this appendix we show that there cannot be a constant $K<N$ for which ${\cal 
D}_N(H,a)\leq K {\cal D}_1(H,a)$ holds for every $H$ and $a$.  Consider the 
family of function, parametrized by $r$, given by
\[
H_r(x)= \eta^4 \exp(-r \eta^2).
\]
Then 
\[
{\cal D}_1(H_r,a) = \sup \frac{H_r(\eta)}{a^2+\eta^2} = 
\frac{H_r(\eta(r))}{a^2+ \eta(r)^2}
\]
for some $\eta(r)$ with $\eta(r)^2\leq \frac{2}{r}$, since 
$H_r(\eta)/(a^2+\eta^2)$ is decreasing on $\eta^2>\frac{2}{r}$.
On the other hand, we get
\[
{\cal D}_N(H_r,a) \geq \frac{N \eta(r)^4 \exp(-r \eta(r)^2) \exp(-\pi(N-1) 
\eta(r)^2)}{a^2 + N \eta(r)^2}
\]
so that
\[
\liminf_{r\to\infty}\frac{{\cal D}_N(H_r,a)}{{\cal D}_1(H_r,a)} \geq 
\liminf_{r\rightarrow \infty} N \frac{a^2+ \eta(r)}{ a^2+ 
N\eta(r)^2}\exp(-\pi(N-1)\eta(r)^2)=N.
\]
This bound is optimal since for any $H$ and $a$ we have
\begin{equation}\label{DnD1A}
{\cal D}_N(H,a) \leq \sup_\eta \frac{\sum_{i=1}^N {\cal D}_1(H,a) 
(a^2+\eta^2)}{ a^2+ N\eta^2} \leq N {\cal D}_1(H,a).
\end{equation}


\medskip
\noindent{\bf Acknowledgements}: Michael Loss and Hagop Tossounian acknowledge partial support from the 
NSF grant DMS-1301555. Michael Loss also acknowledges partial support from the NSF grant DMS- 1600560 and the Humboldt Foundation.

\bibliographystyle{plain} 

\end{document}